\documentclass[11pt]{article}

\usepackage[T1]{fontenc}
\usepackage[utf8]{inputenc}
\usepackage[english]{babel}

\usepackage{amsfonts,amstext,amsmath,amssymb,amsthm} 

\usepackage{algorithm} \usepackage{algpseudocode}

\usepackage{graphicx}
\usepackage{textcomp}
\usepackage{xcolor}

\usepackage{acronym}

\usepackage{booktabs} 

\usepackage{placeins} 

\usepackage{enumitem}

\usepackage[]{subfig}
\usepackage{url}

\usepackage{titlesec}
\usepackage{etoolbox}
\makeatletter
\patchcmd{\ttlh@hang}{\parindent\z@}{\parindent\z@\leavevmode}{}{}
\patchcmd{\ttlh@hang}{\noindent}{}{}{}
\makeatother

\usepackage[authoryear]{natbib}

\titleformat{\section}{\large\bfseries}{\thesection.}{.5em}{}
\titlespacing*{\section}{0pt}{*3}{*2}
\titleformat{\subsection}{\normalfont\bfseries}{\thesubsection.}{.5em}{}
\titlespacing*{\subsection} {0pt}{*3}{*2}
\titleformat{\subsubsection}{\normalfont\bfseries}{\thesubsubsection.}{.5em}{}
\titlespacing*{\subsubsection} {0pt}{*3}{*2}

\usepackage{lmodern}

\long\def\symbolfootnote[#1]#2{\begingroup \def\thefootnote{\fnsymbol{footnote}}\footnote[#1]{#2}\endgroup}

\usepackage[]{float,latexsym,times}

\usepackage{hyperref}
\usepackage[capitalise]{cleveref}

\graphicspath{{./img/}}

\theoremstyle{plain} \newtheorem{theorem}{Theorem}[section]
\newtheorem{lemma}{Lemma}[section]
\newtheorem{corollary}{Corollary}[section]

\theoremstyle{definition}

\crefname{corollary}{Corollary}{Corollaries}
\crefname{problem}{Problem}{Problems}
\crefname{remark}{Remark}{Remarks}
\crefname{appsec}{Appendix}{Appendices}
\crefname{lemma}{Lemma}{Lemmas}
\crefname{enumi}{}{}
\crefname{equation}{}{}

\makeatletter \@ifclassloaded{article}{\numberwithin{equation}{section} \numberwithin{remark}{section}}{
  }\makeatother

\crefname{equation}{}{}

\DeclareMathOperator{\E}{\mathsf{E}} 		\DeclareMathOperator{\Var}{\mathsf{Var}} 	 \DeclareMathOperator*{\argmax}{\text{arg\,max\;}}  

\newcommand{\R}{\mathbb{R}}
\newcommand{\nonNegSet}{\R_{\geq0}}

\newcommand{\Hyp}{\mathrm{H}} 	

\newcommand{\given}{\,|\,}

\newcommand{\bbgiven}{\,\Big|\,}

\newcommand{\param}{\ensuremath{\theta}}
\newcommand{\paramRV}{\ensuremath{\Theta}}
\newcommand{\paramRVSet}{\ensuremath{\Theta}}

\newcommand{\dInt}{\mathrm{d}}	\newcommand{\indd}[1]{\ensuremath{\mathbf{1}_{#1} }} \newcommand{\ind}[1]{\indd{\{#1\}} } \newcommand{\stat}[1]{\ensuremath{\mathbf{t}_{#1}}} \newcommand{\statWOa}{\ensuremath{{\stat{}}}} \newcommand{\transkernel}[1][n]{\ensuremath{\xi_{\stat{#1}}}}

\newcommand{\detConstr}[1][m]{\bar{\alpha}^{#1}}
\newcommand{\estConstr}[1][m]{\bar{\beta}^{#1}}

\newcommand{\updateProbMeasure}[1][n]{\ensuremath{Q_{\stat{#1}}}}
\newcommand{\updateProbMeasureTilde}[1][n]{\ensuremath{\tilde Q_{\stat{#1},\postProbTuple[#1]}}}
\newcommand{\updateProbMeasureTildeScaled}[2][n]{\ensuremath{\tilde Q_{\stat{#1},#2\postProbTuple[#1]^\bullet}}}

\newcommand{\auxVarCostOptTilde}[1]{\ensuremath{{\tilde D^\star_{#1}}}}

\newcommand{\RV}{\ensuremath{X}}
\newcommand{\RVidx}[1][n]{\ensuremath{\RV_{#1}}}
\newcommand{\SeqDataRV}{\ensuremath{\mathcal{X}}}
\newcommand{\obsIdx}[1]{\ensuremath{\boldsymbol{x}_{#1}}}
\newcommand{\obs}{\obsIdx{n}}
\newcommand{\obsScalar}[1]{x_{#1}}
\newcommand{\xnew}{x}	

\newcommand{\stopR}{\ensuremath{\Psi} }

\newcommand{\stopOpt}{\ensuremath{\stopR^\star}}
\newcommand{\dec}{\ensuremath{\delta}}

\newcommand{\decOpt}{\ensuremath{\dec^\star}}

\newcommand{\stopAt}[1][]{\ensuremath{\Phi_{n#1}}}

\newcommand{\costDet}[1][n]{\ensuremath{\lambda}}
\newcommand{\costEst}[1][n]{\ensuremath{\mu}}
\newcommand{\costDetOpt}[1][n]{\ensuremath{\lambda^\star}}
\newcommand{\costEstOpt}[1][n]{\ensuremath{\mu^\star}}

\newcommand{\est}[1]{\ensuremath{\hat\param_{#1}}}

\newcommand{\policy}{\ensuremath{\pi}}
\newcommand{\policySet}{\ensuremath{\Pi}}
\newcommand{\policyOptC}{\ensuremath{\policy_{\costDet,\costEst}^\star}}
\newcommand{\policyOptCerr}{\policy_{\costDet_{\detConstr[],\estConstr[]}^\star,\costEst_{\detConstr[],\estConstr[]}^\star}^\star}
\newcommand{\policyOptErr}{\policy_{\detConstr[],\estConstr[]}^\star}

\newcommand{\policyFull}{\ensuremath{\{\stopR_n, \dec_n, \est{1,n},\ldots,\est{M,n}\} _{0\leq n \leq N}}}
\newcommand{\policyOptFull}{\ensuremath{\{\stopR_n^\star, \dec_n^\star, \est{1,n}^\star,\ldots,\est{M,n}^\star\} _{0\leq n \leq N}}}

\newcommand{\errorDet}[2]{\ensuremath{\alpha_{#1}^{#2}}}
\newcommand{\errorEst}[2]{\ensuremath{\beta_{#1}^{#2}}}

\newcommand{\stopRegionCompl}[1][n]{\ensuremath{\bar{\mathcal{S}}_{#1}}}

\newcommand{\stopRegion}[1][n]{\ensuremath{{\mathcal{S}_{#1}}}}
\newcommand{\stopRegionDec}[2]{\ensuremath{{\mathcal{S}_{#1}^{#2}}}}

\newcommand{\stopRegionBound}[1][n]{\ensuremath{{\partial\mathcal{S}_{#1}}}}
\newcommand{\stopRegionBoundTilde}[1][n]{\ensuremath{{\partial\mathcal{\tilde S}_{#1}}}}
\newcommand{\stopRegionDecMSPRT}[1]{\ensuremath{{\mathcal{S}_{\text{MSPRT}}^{#1}}}}

\newcommand{\Gam}{\mathrm{Gam}}
\newcommand{\unif}{\mathcal{U}}
\newcommand{\norm}[2]{\ensuremath{\mathcal{N}\left(#1,#2\right)}}

\newcommand{\sampleSpaceObs}[1][]{\ensuremath{\Omega_X^{#1}}}
\newcommand{\parameterSpace}{\ensuremath{\mathbb{R}}}
\newcommand{\sampleSpaceHyp}{\ensuremath{\Omega_\Hyp}}

\newcommand{\stateSpaceObs}[1][]{\ensuremath{E_X^{#1}}}
\newcommand{\stateSpaceParam}{\ensuremath{E_{\paramRVSet}}}

\newcommand{\stateSpaceStat}[1][]{\ensuremath{E_\statWOa^{#1}}}
\newcommand{\stateSpaceHyp}[1][]{\ensuremath{E_\Hyp}}

\newcommand{\logLikelihoodRatio}{\ensuremath{\eta}}

\newcommand{\postProbwo}{\ensuremath{e}}
\newcommand{\postProbTuple}[1][n]{\ensuremath{\mathbf{\postProbwo}_{#1}}}
\newcommand{\postProb}[2][n]{\ensuremath{\postProbwo_{#2,#1}}}

\newcommand{\transkernelPostVar}[3]{\ensuremath{\tilde\xi_{#3}(#1, #2)}}

\newcommand{\stateSpacePostProb}{\ensuremath{E_{\postProbTuple[]}}}
\newcommand{\metricPostProb}{\ensuremath{\mathcal{E}_{\postProbTuple[]}}}

\newcommand{\invGam}{\ensuremath{\mathrm{I}\Gam}}

\newcommand{\shape}{\ensuremath{a}}
\newcommand{\scale}{\ensuremath{b}}
\newcommand{\shapePost}{\ensuremath{{\shape_n}}}
\newcommand{\scalePost}{\ensuremath{{\scale_{m,n}}}}
\newcommand{\shapePP}{\ensuremath{{\tilde{\shape}_n}}}
\newcommand{\scalePP}{\ensuremath{{\tilde{\scale}_{m,n}}}}

\newcommand{\scaleFct}{\ensuremath{\gamma}}

\newlength{\imgWidthSingle}
\newlength{\imgWidthDouble}
 \acrodef{AWGN}{additive white Gaussian noise}
\acrodef{ASK}{Amplitude Shift Keying}
\acrodef{iid}{independent and identically distributed}
\acrodef{SNR}{Signal-to-Noise Ratio}
\acrodef{MSE}{Mean-Squared Error}
\acrodef{MAE}{Mean Absolute Error}
\acrodef{SPRT}{Sequential Probability Ratio Test}
\acrodef{MMSE}{Minimum Mean-Squared Error}
\acrodef{MSPRT}{Matrix Sequential Probability Ratio Test}
\acrodef{BPSK}{Binary Phase Shift Keying}
\acrodef{pdf}{probability density function}
\acrodef{LP}{Linear Program} 
\begin{document}
\title{\textbf{\Large   Bayesian Sequential Joint Detection and Estimation under Multiple Hypotheses}}

\date{}

\maketitle

\author{
	\begin{center}
		\vskip -1cm
		\textbf{\large Dominik Reinhard\textsuperscript{1}, Michael Fau\ss{}\textsuperscript{2}, and Abdelhak M. Zoubir\textsuperscript{1}} \\
		\textsuperscript{1} Signal Processing Group, Technische Universit\"at Darmstadt, Darmstadt, Germany\\
		\textsuperscript{2} Department of Electrical Engineering, Princeton University, Princeton, NJ 08544, USA
	\end{center}
}
\symbolfootnote[0]{\normalsize Address correspondence to Dominik Reinhard,
	Signal Processing Group, Technische Universit\"at Darmstadt, Merckstra\ss{}e 25, 
	64283 Darmstadt, Germany; E-mail: reinhard@spg.tu-darmstadt.de
}

\algnewcommand{\Inputs}[1]{\State \textbf{inputs:} #1
}
\algnewcommand{\Initialize}[1]{\State \textbf{initialize:} #1
}

{\small \noindent\textbf{Abstract:}

 We consider the problem of jointly testing multiple hypotheses and estimating a random parameter of the underlying distribution. This problem is investigated in a sequential setup under mild assumptions on the underlying random   process.
 The optimal method minimizes the expected number of  samples while ensuring that the average detection/estimation  errors do not exceed a certain level.
 After converting the constrained problem to an unconstrained one, we characterize the general solution by a non-linear Bellman equation, which is parametrized by a set of cost coefficients. A strong connection between the derivatives of the cost function with respect to the coefficients and the detection/estimation errors of the sequential procedure is derived. 
 Based on this fundamental property, we further show that for suitably chosen cost coefficients the solutions of the constrained and the unconstrained problem coincide.
 We present two approaches to finding the optimal coefficients. For the first approach, the final optimization problem is converted into a linear program, whereas the second approach solves it with a projected gradient ascent.
 To illustrate the theoretical results, we consider two problems for which the optimal schemes are designed numerically. Using Monte Carlo simulations, it is validated that the numerical results agree with the theory.
 }

\vspace*{1em}
{\small \noindent\textbf{Keywords:} 
Joint detection and estimation, Multiple hypotheses; Sequential analysis; Stopping time
}
\\[1em]
{\small \noindent\textbf{Subject Classifications:} 
	62L10; 62L12; 62L15; 90C05; 93E10.

\section{Introduction}
In many applications, hypothesis testing and parameter estimation occur in a coupled way and both are of equal interest. More precisely, one has to decide between two or more hypotheses and, depending on the decision, estimate one or more, possibly random, parameters of the underlying distribution. This problem goes back to Middleton \emph{et al.}, who initially addressed this problem in the late 1960s. In \citep{Middleton1968Simultaneous}, a Bayesian framework was used to find a jointly optimal solution. After extending the problem to multiple hypotheses \citep{fredriksen1972simultaneous}, the interest in the topic declined in the literature, but the topic regained more attention \citep{chaari2013fast,fauss2017sequential,jajamovich2012Minimax,Li2016Optimal,Momeni2015Joint,moustakides2012joint,reinhard2019JointSNR,reinhard2019bayesian,reinhard2016,tajer2010Optimal,Vo2010Joint, jan2018iterative} since 2000. 

Joint detection and estimation is of particular interest in signal processing and communications. For example, in cognitive radio, the secondary user has to detect the primary user and estimate the possible interference such that dynamic spectrum access can be performed \citep{Yilmaz2014Sequential}. In a general communication setup, one is interested in detecting the presence of a signal and estimating the channel\citep{jan2018iterative}.
There exist many more applications, such as radar\citep{tajer2010Optimal}, speech processing \citep{Momeni2015Joint}, change point detection and estimation of the time of change \citep{boutoille2010hybrid}, optical communications \citep{wei2018simultaneous}, detection and estimation of objects from images \citep{Vo2010Joint} and biomedical engineering \citep{chaari2013fast,jajamovich2012Minimax}.

In applications such as radar, where the aim is to detect a single target and to estimate, for example, the velocity of the target, formulating the joint detection and estimation problem with two hypotheses is sufficient.
Nevertheless, there exist a variety of applications where we have to decide between multiple hypotheses and to estimate simultaneously one or more parameters.
In communications, for example, one wants to decode a received signal and estimate some parameters of the signal model \citep{Goldsmith2005Wireless}, e.g., the unit sample response of the channel or the noise power.
If the alphabet consists of more than two symbols, testing multiple hypotheses is required.
Another example is state estimation for smart grids, where topological uncertainties can be modeled using multiple hypotheses \citep{huang2014Adaptive,yilmaz2015sequential}. 

Constraints on the latency of a system arise in many signal processing applications.
Therefore, it is required to perform inference as fast as possible and ensure at the same time its quality.
This leads to the field of sequential analysis that was pioneered by Wald in the late 1940 \citep{wald1947sequential}.
Sequential methods stop sensing new data once one is confident enough about the phenomenon of interest.
Compared to their fixed-sample-size counterparts, these methods require on average significantly less samples while they reach the same inference quality.
An overview on sequential detection and estimation can be found in \citep{tartakovsky2014sequential} and \citep{ghosh2011sequential}, respectively.

Combining the key ideas of sequential analysis with those of joint detection and estimation leads to a powerful framework which is applicable for a wide range of signal processing tasks. In this framework, the average number of used samples should be minimized, while the detection and estimation errors are controlled.

\subsection{State of the Art}
The design of optimal procedures for the problem of sequential joint detection and estimation is an important yet little studied topic.
There exist well established sequential tests in the literature that include an estimation step, such as the generalized \ac{SPRT} or the adaptive \ac{SPRT}.
For those tests, only the outcome of the hypothesis test is of interest and hence, it is not possible to control the quality of the estimate.
Details about the generalized \ac{SPRT} and the adaptive \ac{SPRT} can be found in \citep[Section 5.4]{tartakovsky2014sequential}.

In \citet{birdsall1973sufficient}, the authors investigated the sequential update of sufficient statistics for simultaneous sequential detection and estimation. However, the authors did not provide an optimal solution for the joint detection and estimation problem.  
The problem of sequential detection and state estimation for Markov processes was investigated in \citet{buzzi2006joint, grossi2008sequential} .
Although the quality of the estimator, in that case the probability of a correct estimation, is of primary interest, it does not affect the number of used samples.
Later on, a test that incorporates both uncertainties, the one about the true hypothesis and the true states, into the number of used samples was designed for a similar problem \citep{grossi2009sequential}.

The design of optimal procedures for the problem of sequential joint detection and estimation was initially treated by Y{\i}lmaz \emph{et al.} In \citet{yilmaz2015sequential}, the aim is to decide between two hypotheses and, if the null is rejected, to estimate a random parameter. The approach in \citet{yilmaz2015sequential} was later extended to multiple hypotheses \citep{Yilmaz2016Sequential} and applied to joint spectrum sensing and channel estimation \citep{Yilmaz2014Sequential}. The objective in the work by Y{\i}lmaz \emph{et al.} \citep{Yilmaz2014Sequential, yilmaz2015sequential, Yilmaz2016Sequential} is to design a scheme that minimizes the number of used samples for \emph{every} set of observations while keeping a combined detection/estimation cost function below a certain level. The cost function is a linear combination of detection and estimation errors with combination weights set by the designer. Although the detection and estimation errors can be balanced by varying the combination weights, the main drawback of that approach is the difficulty in designing  procedures with predefined performance in terms of error probabilities and estimation errors.

In \citet{reinhard2016}, we investigated the problem of sequential joint detection and estimation under distributional uncertainties. In \citet{fauss2017sequential}, the problem of sequential joint signal detection and \ac{SNR} estimation was addressed. In that work, the \ac{SNR} was modeled as a deterministic unknown parameter and the detection and estimation errors were kept below predefined limits.
More recently, we developed a Bayesian framework for two hypotheses.
In that work, the \emph{average} number of used samples is minimized while the detection and estimation errors are kept below predefined levels\citep{reinhard2019bayesian}. Contrary to the work by Y{\i}lmaz \emph{et al.} \citep{Yilmaz2014Sequential, yilmaz2015sequential, Yilmaz2016Sequential} which only allows a balance of detection and estimation errors, the framework in \citet{reinhard2019bayesian} comes with an approach to choose the cost coefficients such that constraints on the detection and estimation errors are fulfilled and such that the resulting scheme uses on average as few samples as possible.
The framework was later applied to joint signal detection and \ac{SNR} estimation \citep{reinhard2019JointSNR}.
Sequential joint detection and estimation in distributed sensor networks was addressed in \citet{reinhard2020Distributed} and \citet{reinhard2020DistributedNonGaussian}.

\subsection{Contributions and Outline}
In this work, we show how the framework in \citet{reinhard2019bayesian} can be recast to the general case of multiple hypotheses and as a result we unify the theory so that the work in \citet{reinhard2019bayesian} becomes a special case of this work.
The aim is to design a strictly optimal method, i.e., a method that fulfills constraints on the detection and estimation errors while it uses on average a minimum number of samples.
To the best of our knowledge, the design and analysis of strictly optimal procedures for sequential joint detection and estimation for the multiple hypotheses case has not been addressed in the open literature so far.
Although this problem might look similar to the one in \citet{reinhard2019bayesian}, the multiple hypotheses case differs significantly from its binary counterpart.
These differences are captured in the cost for stopping the scheme and, as a consequence, also in the optimal cost function that characterizes the procedure.
To design the procedure, the coefficients that parameterize the optimal cost function have to be chosen such that the constraints on the error probabilities and the \acp{MSE} are fulfilled.
To this end, a connection between the performance measures, i.e., the detection and estimation errors, and the partial derivatives of the cost function with respect to the corresponding coefficients is exploited as in \citet{reinhard2019bayesian}.
However, due to the structural differences of the optimal cost function in the binary and the multiple hypotheses scenario, a connection between the error probabilities and the partial derivatives of the cost function is a non-obvious one.
The proof of this property is given in this work. Although similar tools as in \citet{reinhard2019bayesian} are used, the two proofs differ technically. 
Based on this connection, the original design problem can be, as its binary counterpart, converted to a linear program.
In addition, we provide a second design algorithm that is based on a projected gradient ascent which was not a subject in \citet{reinhard2019bayesian}.

Note that a small amount of some results in this paper were presented at the 2020 IEEE International Conference on Acoustics, Speech, and Signal Processing (ICASSP)\citep{reinhard2020Multiple}.

This work is structured as follows: in \cref{sec:probForm}, we summarize the underlying assumptions and provide a formal problem formulation. \cref{sec:optStopping} describes the conversion of this problem to an optimal stopping problem. The solution of this problem, which can be obtained by dynamic programming, is characterized by a set of non-linear Bellman equations. The properties of the resulting cost functions are discussed in \cref{sec:propCostFct}. In \cref{sec:optCostCoeff}, we present two approaches for selecting the  coefficients, which parametrize the cost functions, such that predefined detection and estimation performances are achieved and the resulting scheme is of minimum average run-length. To illustrate the proposed approaches, \cref{sec:numResults} provides two numerical examples. 
 \section{Problem Formulation}\label{sec:probForm}
Let $\SeqDataRV= (X_n)_{n \geq 1}$ be a sequence of random variables, which is generated under one out of $M$ different hypotheses $\Hyp_m$, $m=1,\ldots,M$. Under each hypothesis, the distribution of the sequence $\SeqDataRV$ is parametrized by a single scalar and real random parameter $\paramRV_m$ with known distribution.
Although the problem can be formulated with multiple parameters under each hypothesis, for the sake of simplicity we focus on the single parameter case. The work on the multiple parameter case will appear elsewhere.
The tuple $(\RVidx, \paramRVSet, \Hyp)$ is defined on the probability space $(\sampleSpaceObs \times \parameterSpace \times \sampleSpaceHyp, \mathcal{B}(\sampleSpaceObs) \otimes \mathcal{B}(\mathbb{R}) \otimes \mathcal{F}_\Hyp , P)$ in the metric state space $(\stateSpaceObs \times \stateSpaceParam \times \stateSpaceHyp, \mathcal{E}_{X} \times \mathcal{E}_{\paramRVSet} \times \mathcal{E}_\Hyp)$. The density of $P$ with respect to a
suitable reference measure, for example, the standard Lebesgue measure when $\stateSpaceObs$ is Euclidean, is denoted by $p$.
Hence, the $M$ composite hypotheses can be written as
\begin{align*}
 \Hyp_m:\, \SeqDataRV\given\paramRV_m \sim p(\mathbf{x}\given\Hyp_m, \param_m)\,,\,\paramRV_m\sim p(\param_m\given\Hyp_m)\,,
\end{align*}
where $m=1,\ldots,M$ and $\mathbf{x}$ are the observations of the random sequence $\mathcal{X}$.
Note that the underlying hypothesis can, in the most general case, affect both, the distribution of the data given the parameter, as well as the distribution of the parameter.
Since the $M$ hypotheses are composite, we do not only want to decide for the true hypothesis, but we are also interested in estimating the underlying parameter $\param_m$.
Hence, the problem is one of \emph{joint detection and estimation}. Moreover, our aim is to solve this problem in a sequential setup, i.e., we observe the sequence $\SeqDataRV$ sample by sample until we are confident enough about the true hypothesis and the underlying parameter. Thus, the problem becomes one of \emph{sequential joint detection and estimation}. 
The decision maker is restricted to take at most $N$ samples.
That is, irrespective of the certainty about the true hypothesis and the true parameter, the sequential scheme has to stop sampling and to infer the quantities of interest.
For a brief discussion on how to choose $N$ in practice, see \cref{sec:optProblem}.

Before going into the details of the optimization problem, we first summarize the assumptions for the proposed framework. Then, some fundamentals of sequential joint detection and estimation are provided, followed by a discussion of the performance metrics.

\subsection{Assumptions and Notation}\label{sec:assumptions}
\begin{enumerate}[label={\textbf{A\arabic*}}]
\item \label{ass:HparamConst}The true hypothesis and the true parameter do not change during the observation period.
\item The $M$ hypotheses are mutually exclusive.\item \label{ass:suffStat} There exists a sufficient statistic $\stat{n}(\obs)$ in a metric state space $(\stateSpaceStat, \mathcal{E}_{\stat{}})$ such that for all $n=1,\ldots,N$, $m=1,\ldots,M$,
\begin{align*}
p(\obsScalar{n+1}, \param_m, \Hyp_m \given \obs) = p(\obsScalar{n+1}, \param_m, \Hyp_m | \stat{n})
\end{align*}
with some initial statistic $\stat{0}$ and a transition kernel
\begin{align*}
\stat{n+1}(\obsIdx{n+1}) = \xi(\stat{n},\obsScalar{n+1}) =: \transkernel(\obsScalar{n+1})\,.
\end{align*}
\item \label{ass:finiteSOM}The second order moments of the random parameters $\paramRV_m$ conditioned on the hypotheses $\Hyp_m$ exist and are finite:
\begin{align*}
 \E[\paramRV_m^2\given\Hyp_m] < \infty\,,\quad \forall m=1,\ldots,M
\end{align*}

\end{enumerate}
For the sake of compactness, the integration domain is not indicated explicitly in integrals that are taken over the whole domain, e.g.,
\begin{align*}
 \E[X] = \int_{\stateSpaceObs} xp(x)\dInt x \equiv \int xp(x)\dInt x\,. 
\end{align*}
Moreover, the dependency of functions, such as estimators, on the data or the sufficient statistic is dropped for simplicity. The dependency should be clear from the context.
The indicator function of the event $\mathcal{A}$ is denoted by $\ind{\mathcal{A}}$.

\subsection{Fundamentals of Sequential Joint Detection and Estimation}
Since the sequence $\SeqDataRV$ is observed sample by sample, the state space of the observations grows with every time instant. To keep the problem tractable, the observations $\obsIdx{n}$ are replaced by the sufficient statistic $\stat{n}$ as stated in \cref{ass:suffStat}, which serves as a low-dimensional representation of the data. Moreover, its state space is independent of the time instant $n$.

For the joint detection and estimation problem, a decision rule, as well as a set of estimators (one under each hypothesis) have to be found. The decision rule maps the sufficient statistic to a decision in favor of one particular hypothesis.
The estimators, on the other hand, map the sufficient statistic onto a point in the state space of the parameter under each hypothesis. Mathematically, the decision rule and the estimators are defined as
\begin{align*}
 \dec_n& : \stateSpaceStat \mapsto  \Delta_M = \{1,\ldots,M\} \,,\\
 \est{m,n}&: \stateSpaceStat \mapsto \stateSpaceParam\,, \quad m=1,\ldots,M\,.
\end{align*}
The decision rule and the estimators depend on the number of samples $n$, which is indicated by the subscript.
In a sequential framework, the number of used samples is not given \emph{a priori}. Thus, a stopping rule needs to be found for deciding whether to stop sampling or not. Mathematically, the stopping rule is defined as
\begin{align*}
 \stopR_n& : \stateSpaceStat \mapsto \Delta = \{0,1\}\,,
\end{align*}
where $0$ and $1$ stand for continue sampling and stop sampling, respectively. Hence, the run-length of the sequential scheme can be defined as
\begin{align*}
 \tau = \min\,\{n\geq 1: \stopR_n = 1\}\,.
\end{align*}
Note that, since the sopping rule evaluates the sufficient statistic, which is a random variable, the run-length is also a random variable.
Contrary to, e.g., \citet{novikov2009optimal, novikov2009multiple}, where the stopping and decision rules are probabilities instead of hard decisions, there is no need for randomization in this work. The reason will become clear later.
The collection of the decision rule, the estimators and the stopping rule is referred to subsequently as policy, and is defined as
\begin{align*}
 \policy = \policyFull\,.
\end{align*}
The set of all feasible policies is given by
\begin{align*}
 \policySet = \bigl\{ \policy: \stateSpaceStat \mapsto \Delta^N \times \Delta_M^N \times \stateSpaceParam^{MN} \bigr\}\,.
\end{align*}
\subsection{Performance Measures}\label{subsec:perfMeasures}
One common performance measure for sequential inference is the average run-length, i.e., the number of samples which are used on average.
Besides that, there exist different performance measures when dealing with multiple hypotheses, e.g., the probability of accepting $\Hyp_i$ when $\Hyp_m$ is true $P(\{\dec_\tau = i\}\given\Hyp_m)$, $i\neq m$, or the probability of falsely rejecting hypothesis $\Hyp_m$ $P(\{\dec_\tau \neq m \}\given \Hyp_m)$.
In this work, we focus on the latter although an extension to the former should not be too difficult.
As for the estimation performance, there exist many measures such as the \ac{MSE} or the \ac{MAE}. In this work, the \ac{MSE} is used to quantify the estimation performance.
More formally, the error metrics can be written as
\begin{align}
 \errorDet{n}{m}(\stat{n}) & =  \E[\ind{\dec_\tau\neq m}\given\Hyp_m, \stat{n}, \tau \geq n]\,, \label{eq:perfMeasuresDet} \\
 \errorEst{n}{m}(\stat{n}) & =  \E[\ind{\dec_\tau= m}(\param_m - \est{m,\tau})^2\given\Hyp_m, \stat{n}, \tau \geq n]\,, \label{eq:perfMeasuresEst}
\end{align}
where
\begin{align*}
 \bigl\{\dec_\tau=m\bigr\} := \bigl\{\stat{n}\in\stateSpaceStat: \dec_\tau(\stat{n}) = m\bigr\}\,.
\end{align*}
This short hand notation for subspaces of the state space is used throughout the paper.

The quantity $\errorDet{n}{m}(\stat{n})$ denotes the probability that hypothesis $\Hyp_m$ is rejected given that $\Hyp_m$ is true and that the scheme is in state $\stat{n}$ at time $n$.
Similarly, $\errorEst{n}{m}(\stat{n})$ is the \ac{MSE} under the condition that the scheme is in state $\stat{n}$ at time $n$. The estimation error is set to zero if the scheme decides erroneously for a hypothesis.

Alternatively, we can express the performance measures in \cref{eq:perfMeasuresDet} and \cref{eq:perfMeasuresEst} recursively. These definitions, which follow from the Chapman-Kolmogorov backward equations \citep{EoM1990}, are given by
\begin{align} \label{eq:perfMeasuresRecFirst}
    \errorDet{n}{m}(\stat{n}) & =  \stopR_n\ind{\dec_n\neq m}  + (1-\stopR_n) \E[\errorDet{n+1}{m}(\stat{n+1})\given\Hyp_m, \stat{n}]\,,\\
    \errorDet{N}{m}(\stat{N}) & =  \stopR_N\ind{\dec_N\neq m}\,, \\
    \errorEst{n}{m}(\stat{n}) & =  \stopR_n\ind{\dec_n= m}\E[(\paramRV_m - \est{m,n})^2\given\Hyp_m,\stat{n}]  + (1-\stopR_n)\E[\errorEst{n+1}{m}(\stat{n+1})\given\Hyp_m, \stat{n}]\,, \\
    \errorEst{N}{m}(\stat{N}) & =  \stopR_N\ind{\dec_N= m}\E[(\paramRV_m - \est{m,N})^2\given\Hyp_m,\stat{N}]\,.
    \label{eq:perfMeasuresRecLast}
\end{align}
The importance of these recursive definitions will become clearer later. For the sake of compactness, the following short-hand notations will be used throughout the paper
\begin{align*}
 \errorDet{}{m} := \errorDet{0}{m}(\stat{0})\,,\quad \errorEst{}{m} := \errorEst{0}{m}(\stat{0})\,.
\end{align*}

\subsection{Formulation as an Optimization Problem}\label{sec:optProblem}
The aim is to design a sequential method which is of minimum expected run-length while the detection and the estimation errors do not exceed certain levels.
As mentioned earlier, the decision maker is allowed to observe at most $N$ samples, i.e., the sequential scheme is truncated. The maximum number of samples might be determined by the application, e.g., due to latency constraints.
If the choice is left to the designer, it should be chosen as large as possible given the computational resources.
However, if $N$ is chosen too small, i.e., the constraints on the detection and/or the estimation errors cannot be fulfilled, the design problem does not have a feasible solution and the algorithm(s) presented in \cref{sec:optCostCoeff} will not converge.

Formally, the problem can be written as the following optimization
\begin{align}\label{eq:constrProblem}
\begin{split}
 & \min_{\policy\in\policySet}\,\E[\tau]\,,\;\stopR_N = 1\,,\\
 \text{subject to}\quad& \errorDet{}{m} \leq \detConstr\,,\quad  \errorEst{}{m} \leq \estConstr\,, \quad m=1,\ldots,M\,,
 \end{split}
\end{align}
where the upper limits on the detection and estimation errors are denoted by $\detConstr$ and $\estConstr$, respectively.
Instead of solving \cref{eq:constrProblem} directly, we first consider the following unconstrained problem

\begin{align}\label{eq:unconstrProb}
 \min_{\policy\in\policySet}\, \left\{\E[\tau] + \sum_{m=1}^M\left( \costDet_m p(\Hyp_m)\errorDet{}{m} + \costEst_mp(\Hyp_m)\errorEst{}{m} \right)\right\}\,,
\end{align}
where $\costDet_m$ and $\costEst_m$, $m=1,\ldots,M$, are some non-negative and finite cost coefficients. Although the prior probabilities $p(\Hyp_m)$ are not subject to optimization, they are included to simplify notation. It is shown in \cref{sec:optCostCoeff} that for suitable chosen $\costDet_m$ and $\costEst_m$, the solutions of \cref{eq:unconstrProb} and \cref{eq:constrProblem} coincide.
 \section{Reduction to an Optimal Stopping Problem}\label{sec:optStopping}
In order to solve the problem in \cref{eq:unconstrProb}, it first has to be converted to an optimal stopping problem. Optimal stopping theory is a framework which is widely used when designing sequential schemes. In that framework, a stopping rule has to be found which trades-off the expected run-length and the inference quality.  More details about optimal stopping and its applications are given in \citet{peskir2006optimal}.
To derive an optimal stopping problem, we proceed as in \citet{reinhard2019bayesian} and optimize \cref{eq:unconstrProb} first with respect to the decision rule and then with respect to the estimators. To do so, the two summands in \cref{eq:unconstrProb}, the expected run-length and the combined detection and estimation errors have to be reformulated. 
The expected run-length can be expressed as
\begin{align*}
 \E[\tau] & = \sum_{n=1}^N \sum_{m=1}^M \int\int n \stopAt p(\Hyp_m, \param_m, \obs)\dInt\param_m\dInt\obs \\
	  & = \sum_{n=1}^N \int n \stopAt p(\obs)\dInt\obs\,,
\end{align*}
where we used the short-hand notation
\begin{align*}
 \stopAt := \stopR_n\prod_{k=0}^{n-1}(1-\stopR_k)\,.
\end{align*}

Similarly, the weighted sum of detection and estimation errors in \cref{eq:unconstrProb} can be simplified to
\begin{align}
 \sum_{m=1}^M\left( \costDet_m p(\Hyp_m)\errorDet{}{m} + \costEst_mp(\Hyp_m)\errorEst{}{m} \right) = &\sum_{n=0}^N \int \sum_{m=1}^M  \stopAt \ind{\dec_n=m} D_{m,n}(\stat{n}(\obs)) p(\obs)\dInt\obs \label{eq:sumOfD}\,.
\end{align}
The cost for stopping at time $n$ and deciding in favor of $\Hyp_m$, $D_{m,n}(\stat{n})$, is given by
\begin{align*}
D_{m,n}(\stat{n}) & = \costEst_m p(\Hyp_m\given\stat{n})\int(\param_m -  \est{m,n})^2p(\param_m\given\Hyp_m,\stat{n})\dInt\param_m + \sum_{i=1,i\neq m}^M \costDet_i p(\Hyp_i\given\stat{n})\,.
\end{align*}
Hence, the overall objective becomes
\begin{align}\label{eq:minWeightedSumFinal}
 \min_{\policy\in\policySet}\, \sum_{n=0}^N \int \stopAt\biggl(n\!+\!\!\sum_{m=1}^M \ind{\dec_n=m} D_{m,n} \biggr)p(\obs)\dInt\obs\,.
\end{align}In order to solve \cref{eq:minWeightedSumFinal}, it is first minimized with respect to the decision rule and then with respect to the estimators.
Since $D_{m,n}(\stat{n})$ is non-negative for all $n=0,\ldots,N, m=1,\ldots,M$, and all $\stat{n}\in\stateSpaceStat$, it holds that \citep[Lemma 1]{novikov2009multiple}
\begin{align*}
\int \sum_{m=1}^M   \ind{\dec_n=m} D_{m,n}p(\obs)\dInt\obs \geq \int \min_{m}D_{m,n} p(\obs)\dInt\obs\,,
\end{align*}
where equality holds if and only if
\begin{align}\label{eq:optDecRuleMult}
 \decOpt_n(\stat{n}) \in\Bigl\{m: D_{m,n}(\stat{n}) = \min_{1\leq i \leq M} D_{i,n}(\stat{n})\Bigr\} \,.
\end{align}In general, more than one hypothesis may fulfill \cref{eq:optDecRuleMult} which then calls for randomizing $\decOpt_n$. In this work, there is no need for randomization. The reason is the same as that stated in \citet[Remark 3.1]{reinhard2019bayesian}.
After we derived the optimal decision rule, the optimal estimators have to be found. Since the optimal sequential estimator is independent of the stopping time \cite[Theorem 5.2.1.]{ghosh2011sequential} and each $D_{m,n}$ only depends on one estimator, we can minimize all $D_{m,n}$ with respect to each estimator $\est{m,n}$ separately.
The estimator which minimizes the \ac{MSE}, i.e., the \ac{MMSE} estimator, is \begin{align}\label{eq:optEst}
 \est{m,n}^\star = \E[\paramRV_m\given\Hyp_m,\stat{n}]
\end{align}and the \ac{MMSE} is given by the posterior variance \citep{levy2008principles}. Hence, the cost for stopping and deciding in favor of $\Hyp_m$ at time $n$ becomes
\begin{align*}
 \begin{split}
 D_{m,n}^\star(\stat{n}) & = \costEst_m p(\Hyp_m\given\stat{n}) \Var[\paramRV_m\given\Hyp_m, \stat{n}]  + \sum_{i=1,i\neq m}^M \costDet_i p(\Hyp_i\given\stat{n})\,.
\end{split}
\end{align*}The optimal stopping rule can now be found by solving the following optimization problem
\begin{align}\label{eq:optStopping}
 \min_{\{\stopR_n\}_{n=0}^N} \sum_{n=0}^N \E[\stopAt(n + g(\stat{n}))]\,,
\end{align}
where
\begin{align}\label{eq:defG}
 g(\stat{n}) = \min\{D^\star_{1,n}(\stat{n})\,,\,\ldots\,,\,D^\star_{M,n}(\stat{n})\}
\end{align}
is the cost for stopping at time $n$. The solution of the optimization problem is fixed in the following theorem.
\begin{theorem}\label{theo:optimalStoppingSol} The solution of \cref{eq:optStopping} is characterized by the non-linear Bellman equations
 \begin{align*}
    \rho_n(\stat{n}) & = \min\{g(\stat{n}), d_n(\stat{n})\}\quad n<N \,,\\
    \rho_N(\stat{N}) & = g(\stat{N})\,,
 \end{align*}
 with $g(\stat{n})$ defined in \cref{eq:defG} and the cost for continuing is
 \begin{align}\label{eq:consCont}
  d_n(\stat{n}) = 1 + \int \rho_{n+1}(\transkernel(\xnew))p(\xnew\given\stat{n})\dInt\xnew\,.
 \end{align}
\end{theorem}
Since the proof of \cref{theo:optimalStoppingSol} does not differ from the ones in the literature, we refer to, e.g., \citep[Appendix A]{reinhard2019bayesian},\citep[Appendix A]{fauss2015linear}.
With a change in measure, \cref{eq:consCont} becomes
\begin{align*}
 d_n(\stat{n}) = 1 + \int \rho_{n+1}\dInt\updateProbMeasure\,,
\end{align*}where
\begin{align}\label{eq:updateProbMeasure}
 \updateProbMeasure(B) := P\left(\left\{\xnew\in\stateSpaceObs: \transkernel(\xnew) \in B\right\}\given\stat{n}\right)
\end{align}for all elements $B$ of the Borel $\sigma$-algebra on $\stateSpaceStat$. The probability measure $\updateProbMeasure^m$, which will be needed later, is defined analogously, but with $P(\cdot\given\stat{n})$ replaced by $P(\cdot\given\Hyp_m,\stat{n})$.
\begin{corollary}\label{corr:integrable}
 Let $\costDet_m$ and $\costEst_m$ be finite for all $m=1,\ldots,M$, then $\rho_{n+1}$ is $\updateProbMeasure$-integrable for all $\stat{n}$ in $\stateSpaceStat$ and all $0\leq n<N$.
\end{corollary}
A proof of \cref{corr:integrable} is given in \cref{sec:proofIntegrable}.
The optimal policy is summarized in the following corollary.
\begin{corollary}\label{corr:optPolicy}
 The optimal policy which solves \cref{eq:unconstrProb} is 
 \begin{align}
  \policyOptC = \policyOptFull\,,
 \end{align}
with $\decOpt_n$ defined in \cref{eq:optDecRuleMult} and $\est{m,n}^\star$ defined in \cref{eq:optEst}. The optimal stopping rule $\stopOpt_n$ is given by
\begin{align*}
 \stopOpt_n(\stat{n}) = \ind{g(\stat{n}) = \rho_n(\stat{n})}\,.
\end{align*}
\end{corollary}
For the optimal policy stated in \cref{corr:optPolicy}, the stopping region of the scheme, its complement and its boundary are
\begin{align}\label{eq:regions}
\begin{split}
 \stopRegion & = \{\stat{n} \in \stateSpaceStat: g(\stat{n}) < d_n(\stat{n})\} \\
 \stopRegionBound & = \{\stat{n} \in \stateSpaceStat: g(\stat{n}) = d_n(\stat{n})\} \\
 \stopRegionCompl & = \{\stat{n} \in \stateSpaceStat: g(\stat{n}) > d_n(\stat{n})\}
 \end{split}
\end{align}
for $n<N$.
Moreover, let
\begin{align}
\begin{split}
 \stopRegionDec{n}{m} = \stopRegion \cup \{\stat{n}\in\stateSpaceStat: \dec_n(\stat{n}) = m\} \\
 \stopRegionDec{n}{\bar m} = \stopRegion \cup \{\stat{n}\in\stateSpaceStat: \dec_n(\stat{n}) \neq m\}
 \end{split}
\end{align}
denote the regions in which the procedure stops and accepts/rejects hypothesis $\Hyp_m$.
 \section{Properties of the Cost Function}\label{sec:propCostFct}
In this section, we present, similarly to \citet{reinhard2019bayesian,fauss2015linear}, the fundamental properties of the cost functions obtained in the last section. However, as mentioned before, these properties do not follow directly from their two-hypotheses counterpart.
These properties are used later to get the optimal cost coefficients such that the solution of \cref{eq:constrProblem} also solves \cref{eq:unconstrProb}. In order to simplify the upcoming derivations, it is important to show that the boundary of the stopping region is a P-null set. This is stated in the following lemma.
\begin{lemma}\label{lem:BoundPnull}
 If the posterior probabilities $p(\Hyp_m\given\stat{n})$, $m=1,\ldots,M$, are continuous random variables with respect to $\stat{n}$, it holds that
 \begin{align*}
  \updateProbMeasure(\stopRegionBound) = 0\,,\quad \forall \stat{n}\in\stateSpaceStat,\; 0\leq n < N\,,
 \end{align*}
i.e., the boundary of the stopping region $\stopRegionBound$ is a P-null set.
\end{lemma}
The proof of \cref{lem:BoundPnull} is laid down in the supplement \cref{sec:proofBoundPnull}. \cref{lem:BoundPnull} implies that the cost minimizing stopping rule is unambiguous so that there is no need for randomization. Before the main properties of the cost functions are presented, the following short hand notations are introduced
\begin{align*}
 z_n^m & := \frac{p(\stat{n}\given\Hyp_m)}{p(\stat{n})}\,,\\
 \{\transkernel \in \stopRegionCompl[n+1]\} & := \{\xnew \in\stateSpaceObs: \transkernel(\xnew)\in\stopRegionCompl[n+1]\}\,.
\end{align*}
\begin{lemma}\label{lem:derivatives}
Let $\rho_{n,\costDet_m}^\prime$ and $\rho_{n,\costEst_m}^\prime$ denote the derivatives of $\rho_n$ with respect to $\costDet_m$ and $\costEst_m$ for $m=1,\ldots,M$, respectively. For $n<N$, it holds that
 \begin{align*}
   \rho_{n,\costDet_m}^\prime(\stat{n})&  =  \ind{\stopRegionDec{n}{\bar m}} p(\Hyp_m\given\stat{n})  + \ind{\stopRegionCompl}\biggl( p(\Hyp_m)z_{n}^m\updateProbMeasure^m\bigl(\stopRegionDec{n+1}{\bar m}\bigr) + \int_{\{\transkernel\in\stopRegionCompl[n+1]\}} \rho^\prime_{n+1,\costDet_m}\dInt\updateProbMeasure      \biggr)
 \end{align*}
 and
 \begin{align*}
   \rho_{n,\costEst_m}^\prime(\stat{n}) & = \ind{\stopRegionDec{n}{m}} p(\Hyp_m) z_n^m \Var[\paramRV_m\given\Hyp_m, \stat{n}] + \ind{\stopRegionCompl}r_n^m
 \end{align*}
 with $r_n^m$ being recursively defined via
 \begin{align*}
  r^m_n & = p(\Hyp_m) z_n^m \!\!\!\!\!\!\int\displaylimits_{\{\transkernel\in\stopRegionDec{n+1}{m}\}}\!\!\!\!\!\Var[\paramRV_m\given\Hyp_m,\transkernel(\xnew)]p(\xnew\given\stat{n},\Hyp_m)\dInt\xnew + \int_{\{\transkernel\in\stopRegionCompl[n+1]\}} \rho^\prime_{n+1,\costEst_m} \dInt\updateProbMeasure\,.
 \end{align*}
For $n=N$, it further holds that
\begin{align*}
 \rho_{N,\costDet_m}^\prime(\stat{N}) & = \ind{\stopRegionDec{N}{\bar m}}p(\Hyp_m)z_N^m \\
   \rho_{N,\costEst_m}^\prime(\stat{N}) & = \ind{\stopRegionDec{N}{m}} p(\Hyp_m) z_N^m \Var[\paramRV_m\given\Hyp_m, \stat{N}]\,.      
\end{align*}
\end{lemma}
The connection of the derivatives with respect to $\costDet_m$ do not follow from \citet{reinhard2019bayesian} straightforwardly. Therefore, the proof is given in \cref{app:proofTheoDerivatives}. For the derivative with respect to $\costEst_m$, we refer to \citet[Appendix D]{reinhard2019bayesian}.
\cref{lem:derivatives} is an intermediate result, which is required to prove the more important result stated in \cref{theo:derivativesPerformanceMeasures}.
Next, we provide a connection between the derivatives of the cost functions stated in the previous theorem and the performance measures stated in \cref{subsec:perfMeasures}. This connection forms the basis to obtain the optimal cost coefficients, as presented in \cref{sec:optCostCoeff}.
\begin{theorem}\label{theo:derivativesPerformanceMeasures}
Let $\rho_{n,\costDet_m}^\prime$ and $\rho_{n,\costEst_m}^\prime$ be as defined in \cref{lem:derivatives}. Then, using the optimal policy given in \cref{corr:optPolicy}, it holds that 
 \begin{align*}
  \rho_{n,\costDet_m}^\prime(\stat{n}) & = p(\Hyp_m)z_n^m \errorDet{n}{m}(\stat{n})\,,\;m=1,\ldots,M\,,\\
  \rho_{n,\costEst_m}^\prime(\stat{n}) & = p(\Hyp_m)z_n^m \errorEst{n}{m}(\stat{n})\,,\;m=1,\ldots,M\,,
 \end{align*}
 and in particular
 \begin{align*}
  \rho_{0,\costDet_m}^\prime(\stat{0}) & = p(\Hyp_m)\errorDet{0}{m}(\stat{0})\,,\;m=1,\ldots,M\,,\\
  \rho_{0,\costEst_m}^\prime(\stat{0}) & = p(\Hyp_m)\errorEst{0}{m}(\stat{0})\,,\;m=1,\ldots,M\,.
 \end{align*}
\end{theorem}
As the connection for the derivatives with respect to $\costDet_m$ and the detection errors do not follow from its two-hypotheses counterpart, the proof is outlined in  \cref{app:proofDerivativesPerformanceMeasures}. For the derivatives with respect to $\costEst_m$, we refer to \citet[Appendix E]{reinhard2019bayesian}.
 \section{Choice of the Optimal Cost Coefficients}\label{sec:optCostCoeff}
This section extends Section 5 in \citet{reinhard2019bayesian} to multiple hypotheses.
Proofs are given for results that do not follow directly from \citet{reinhard2019bayesian}.

Often, the question how to choose the coefficients of a cost function stays untouched and the choice is left to the designer \citep{Yilmaz2016Sequential, novikov2009multiple}. However, in this work, our aim is to design a sequential scheme which is of minimum average run-length and which fulfills predefined constraints on the detection and estimation errors. Since all three performance measures, average run-length, error probabilities and \ac{MSE}, are of different numerical range, it is rather impossible to choose the cost coefficients by hand.
In order to automatically select the correct cost coefficients such that a predefined performance is achieved, the results stated in \cref{sec:propCostFct} are exploited.

For the sake of a compact notation, let $\lambda=(\lambda_1,\ldots,\lambda_M)$ and $\mu=(\mu_1,\ldots,\mu_M)$. To obtain the set of optimal cost coefficients, we consider the following maximization problem, which is in fact the Lagrangian dual of \cref{eq:constrProblem},
\begin{align} \label{eq:dualProblem}
 \max_{\costDet\geq0,\costEst\geq0}\, & L_{\detConstr[],\estConstr[]}(\costDet,\costEst)\,,
\end{align}
where $\costDet\geq0$, $\costEst\geq0$ have to be read element-wise and 
\begin{align}
 L_{\detConstr[],\estConstr[]}(\costDet,\costEst) = \rho_0(\stat{0}) - \sum_{m=1}^M p(\Hyp_m)(\costDet_m\detConstr + \costEst_m\estConstr)\,.
\end{align}
As it is not trivial to see that strong duality holds and therefore the solutions of \cref{eq:constrProblem} and \cref{eq:dualProblem} coincide, it is fixed in the following theorem.
\begin{theorem}\label{theo:optCostCoeff}
Let $\policyOptErr$ be the solution of \cref{eq:constrProblem}, let $\costDetOpt_{\detConstr[],\estConstr[]}$ and $\costEstOpt_{\detConstr[],\estConstr[]}$  be solutions of \cref{eq:dualProblem} and let $\policyOptCerr$ be the policy parametrized by $\costDetOpt_{\detConstr[],\estConstr[]}$, $\costEstOpt_{\detConstr[],\estConstr[]}$. Then, it holds that
\begin{align*}
 \policyOptErr & = \policyOptCerr\,, \\
 L_{\detConstr[],\estConstr[]}(\costDetOpt_{\detConstr[],\estConstr[]}, \costEstOpt_{\detConstr[],\estConstr[]}) & = \E\biggl[\tau\bbgiven\policy=\policyOptCerr\biggr]\,.
\end{align*}
That is, a solution of \cref{eq:constrProblem} also solves \cref{eq:dualProblem}. Moreover, the optimal objective of \cref{eq:dualProblem} is the expected run-length.
\end{theorem}
A proof of \cref{theo:optCostCoeff} is outlined in \cref{sec:proofOptCoeff}.

Hence, by using \cref{theo:optCostCoeff} and \cref{eq:dualProblem} the original problem given in \cref{eq:constrProblem} is equivalent to
\begin{align} \label{eq:lagrFinal}
\max_{\costDet\geq0,\costEst\geq0} & \,\left\{ \rho_0(\stat{0}) \!-\!\sum_{m=1}^M p(\Hyp_m)(\costDet_m\detConstr + \costEst_m\estConstr)\right\}\,, \\
 \text{s.t.}\; & \rho_n(\stat{n}) = \min\{g(\stat{n}), d_n(\stat{n})\}\,,\quad n<N\,, \nonumber\\
		  & \rho_N(\stat{N}) = g(\stat{N})\,.  \nonumber
\end{align}
In what follows, we present two approaches to solve \cref{eq:lagrFinal}. The first one uses linear programming to obtain the optimal cost coefficients and the cost functions, whereas the second one uses a projected gradient ascent.

\subsection{Linear Programming}
The first approach solves \cref{eq:lagrFinal} by linear programming. To this end, we proceed similarly to \citet{fauss2015linear, reinhard2019bayesian} and relax the equality constraints in \cref{eq:lagrFinal} to multiple inequality constraints and add the cost functions to the set of free variables, i.e.,
\begin{align}\label{eq:LP}
\max_{\substack{\costDet\geq0,\costEst\geq0\\\rho_n\in\mathcal{L}}}
 \, & \biggl\{\rho_0(\stat{0}) - \sum_{m=1}^Mp(\Hyp_m)(\costDet_m\detConstr + \costEst_m\estConstr)\biggr\}\,,\\
 \text{s.t.} \quad
 &\rho_n \leq D^\star_{m,n}\;,\; m=1,\ldots,M,\; n=0,\ldots,N\,, \nonumber\\
 &\rho_n \leq 1+\int\rho_{n+1}\dInt\updateProbMeasure\,, \; n=0,\ldots,N-1\,,\nonumber
\end{align}
where $\mathcal{L}$ is the set of all $\updateProbMeasure$-integrable functions on $\stateSpaceStat$.
\begin{theorem}
 Problem \cref{eq:constrProblem} is equivalent to Problem \cref{eq:LP}.
\end{theorem}
For a proof see \cite{reinhard2019bayesian}.
\subsection{Projected Gradient Ascent}
Besides the previously presented \ac{LP}, we present a second approach that is not as sensitive to the size of the design problem as the \ac{LP}. This approach uses a projected gradient ascent to obtain the optimal cost coefficients and is summarized in \cref{alg:gradAsc}.

At each iteration, the cost functions, as stated in \cref{theo:optimalStoppingSol}, are calculated for a given set of cost coefficients. Next, given the policy induced by these cost functions, the gradient of the objective in \cref{eq:lagrFinal} has to be obtained, which is given by
\begin{align}\label{eq:gradLex}
\begin{split}
 \nabla_{\costDet} L_{\detConstr[],\estConstr[]}(\costDet, \costEst) & = \bigl[p(\Hyp_1)(\errorDet{}{1}\! -\! \detConstr[1]),\ldots, p(\Hyp_M)(\errorDet{}{M}\!-\! \detConstr[M]) \bigr]\,,\\
 \nabla_{\costEst} L_{\detConstr[],\estConstr[]}(\costDet, \costEst) & = \bigl[p(\Hyp_1)(\errorEst{}{1}\! -\! \estConstr[1]),\ldots, p(\Hyp_M)(\errorEst{}{M}\!-\!\estConstr[M]) \bigr]\,.
 \end{split}
\end{align}
This gradient can, e.g., be calculated based on the definitions of the performance measures as stated in \crefrange{eq:perfMeasuresRecFirst}{eq:perfMeasuresRecLast}.
To update the cost coefficients, the old coefficients are shifted in the direction of the gradient and then projected onto the set of feasible coefficients, i.e., the set of non-negative reals.
To control the convergence speed of the algorithm, the gradient is scaled by a factor $\scaleFct$.
These steps are repeated until the solution converges to an optimum. Recall from \cref{sec:proofOptCoeff} that the optimality criteria for the cost coefficients are
\begin{align}\label{eq:optCritCoeff}
 \begin{split}
  \costDet_m(\errorDet{}{m} - \detConstr) & = 0\,,\; m=1,\ldots,M\,,\\
  \costEst_m(\errorEst{}{m} - \estConstr) & = 0\,,\; m=1,\ldots,M\,.
 \end{split}
\end{align}
Hence, the procedure has to be repeated until \cref{eq:optCritCoeff} holds (approximately).

Since the calculation of the performance measures based on their recursive definition can become numerically unstable and hence lead to an inaccurate gradient, a modification of the aforementioned method can be used. Similarly to \citet{reinhard2019JointSNR}, the detection and estimation errors in \cref{eq:gradLex} and \cref{eq:optCritCoeff} can be replaced by their Monte Carlo estimates.

As Problem \cref{eq:lagrFinal} is convex, the projected gradient ascent converges to a global optimum irrespective of the choice of the starting point. Nevertheless, the choice of the starting point is crucial for the convergence speed. To have a fast convergence, we suggest to solve \cref{eq:LP} on a coarse grid if possible and then run \cref{alg:gradAsc} on a finer grid to obtain the optimal cost coefficients.

\begin{algorithm}[!t]
  \newlength{\mynegid}
  \setlength{\mynegid}{-.25cm}
  \begin{algorithmic}[1]
    \Inputs{$\detConstr[1],\ldots, \detConstr[M], \estConstr[1],\ldots,\estConstr[M],\costDet^{(0)}, \costEst^{(0)}, \scaleFct$}
    \Initialize{Set $k\gets 0$}
    \Repeat
      \State \hspace{\mynegid}Set $k \gets k+1$
      \State \hspace{\mynegid}Get policy from \cref{corr:optPolicy} using $\costDet^{(k-1)},\costEst^{(k-1)}$
      \State \hspace{\mynegid}Get gradients from \cref{eq:gradLex} \State \hspace{\mynegid}Set $\costDet^{(k)} = \max\{\costDet^{(k-1)} + \scaleFct\nabla_{\costDet} L_{\detConstr[],\estConstr[]}(\costDet^{(k-1)}, \costEst^{(k-1)}),0\}$
      \State \hspace{\mynegid}Set $\costEst^{(k)} = \max\{\costEst^{(k-1)} + \scaleFct\nabla_{\costEst} L_{\detConstr[],\estConstr[]}(\costDet^{(k-1)}, \costEst^{(k-1)}),0\}$
    \Until{\cref{eq:optCritCoeff} holds approximately.}
    \State \Return{$\costDet^{(k)}, \costEst^{(k)}$}
  \end{algorithmic}
  \caption{Projected Gradient Ascent}
  \label{alg:gradAsc}
\end{algorithm}

\subsection{Discussion}
Although the algorithms presented in this section are in general problem independent, their tractability highly depends on the size of the problem.
With increasing dimensionality of the sufficient statistic, the discretization of its space and hence, the calculation of the optimal policy becomes infeasible.
However, even if the observations are high-dimensional, it might still be possible to find a low-dimensional sufficient statistic.
An example for such an application could be a sensor network in which all sensors send their observations to a fusion center that performs the inference.

Besides the dimensionality of the sufficient statistic, also the dimensionality of the observations may play a role in the implementation of the algorithms.
The cost for continuing, i.e., $d_n(\stat{n}) = 1 + \E[\rho_{n+1}(\stat{n+1})\given\stat{n}]$, can be 
calculated numerically in a straightforward manner for the case of scalar observations. For more dimensional observations, approximating the integral by Markov Chain Monte Carlo integration, as, e.g., in \citet{reinhard2020Distributed}, might be more suitable.

Nevertheless, to overcome this shortcoming, more sophisticated design approaches that do not require a discretization of the state space have to be found.
We have ongoing work on machine learning techniques to learn the policy from observed trajectories rather than calculating the cost function on a discretized state space.

 \section{Numerical Results}\label{sec:numResults}
In this section, we provide two numerical examples to illustrate and validate the proposed approach. First, a simple example is presented to illustrate the basic properties of the optimal method. The second example is more complex and shows how to apply the proposed method to real-life applications.

In order to solve the linear program in \cref{eq:LP}, the continuous spaces are first discretized. The discrete linear program is then solved by the Gurobi optimizer \cite{gurobi} which is called via the MATLAB cvx interface \cite{cvx,gb08}.
We add a regularization term to the objective in \cref{eq:LP} to ensure numerical stability.
See, e.g., \citet[Appendix G]{reinhard2019bayesian} for details. To validate the performance of the designed schemes, a Monte Carlo simulation is performed for both examples. 

\subsection{Benchmarking Method}
In order to compare the proposed approach with existing methods, we choose a two-step procedure as benchmarking method, namely, a standard sequential detector for multiple hypotheses followed by an \ac{MMSE} estimator.
Although there exist different sequential detectors for the multiple hypotheses case, we resort to the \ac{MSPRT}\citep{tartakovsky2014sequential,tartakovsky1998asymptotic} due to its easy implementation and favorable theoretical properties. Alternative sequential detectors for multiple hypotheses can be found in, e.g., \citet{baum1994sequential, tartakovsky2014sequential}.
This two-step procedure is not optimal for the joint detection and estimation problem, but the \ac{MSPRT} is asymptotically optimal for the detection part and the \ac{MMSE} estimator is the optimal estimator with respect to the \ac{MSE}.
For the \ac{MSPRT}, the pair-wise log-likelihood ratios for hypotheses $\Hyp_m$ and $\Hyp_j$ are used and are defined as
\begin{align*}
  \logLikelihoodRatio_{mj}(\stat{n}) = \log\biggl(\frac{p(\stat{n}\given\Hyp_m)}{p(\stat{n}\given\Hyp_j)}\biggr)\,, \quad 
m,j =1,\ldots,M\,,\;\; m  \neq j\,.
\end{align*}
In general, the stopping and decision rules of the \ac{MSPRT} are given by\citep[Eqs. (4.3) and (4.4)]{tartakovsky2014sequential}:
\begin{align*}
   \stopR_n^\text{MSPRT} & =\begin{cases}
					  1 & \exists m: \logLikelihoodRatio_{mj} \geq A_{mj},\;\forall j\in\{1,\ldots,M\}\setminus m\,,\\
                                          0 & \text{else}\,,\\
                                       \end{cases}\\
   \dec_n^\text{MSPRT} & = \biggl\{
					  m: \logLikelihoodRatio_{mj} \geq A_{mj},\;\forall j\in\{1,\ldots,M\}\setminus m
                                       \biggr\}\,.
\end{align*}The thresholds $A_{mj}$, which are used in the stopping and decision rules, have now to be set such that the target error probabilities are met.
To keep the probabilities of falsely rejecting hypothesis $\Hyp_m$ under a certain level $\detConstr$, $m=1,\ldots,M$, the thresholds have to be calculated as \citep[Eq. (4.4)]{tartakovsky2014sequential}
\begin{align*}
 A_{mj} = A_m\approx \log(M/\detConstr)\,.
\end{align*}
Since the presented optimal sequential scheme is a truncated one, we use a truncated two-step procedure for the sake of fair comparison. Hence, the stopping and decision rules at the truncation point are given by
\begin{align*}
   \stopR_N^\text{MSPRT} & = 1\,\quad\text{and}\quad\dec_N^\text{MSPRT} =  \argmax_m \sum_{j=1,j\neq m}^M \logLikelihoodRatio_{mj}\,.
\end{align*}

\subsection{Shift-in-Mean}
The first numerical example is used to show the basic properties of the optimal sequential scheme. Here, we consider three different hypotheses with a Gaussian likelihood. Under each hypothesis, the mean follows a different distribution, whereas the variances of the likelihood are equal. The three different hypotheses are given by
\begin{align*}
  \begin{split}
   \Hyp_1:&\; x_n\given\mu_1 \overset{\text{\ac{iid}}}{\sim} \norm{\mu_1}{\sigma^2}\,, -\mu_1 + 1.3 \sim \Gam(1.7,1)\,,\\ \Hyp_2:&\; x_n\given\mu_2 \overset{\text{\ac{iid}}}{\sim} \norm{\mu_2}{\sigma^2}\,, \phantom{- + 1.3}\mu_2 \sim \unif(-1,1)\,,\\ \Hyp_3:&\; x_n\given\mu_3 \overset{\text{\ac{iid}}}{\sim} \norm{\mu_3}{\sigma^2}\,, \phantom{-}\mu_3 - 1.3 \sim \Gam(1.7,1)\,, \end{split}
\end{align*}
where $\norm{\mu}{\sigma^2}$ is the normal distribution with mean $\mu$ and variance $\sigma^2$, $\unif(l,u)$ is the uniform distribution on the interval $[l,u)$ and $\Gam(a,b)$ is the Gamma distribution with shape and scale parameters $a$ and $b$, respectively.
All three hypotheses have equal prior probabilities and the variance is set to $\sigma^2=4$.
The aim is to design an optimal sequential scheme to simultaneously test the three hypotheses and to estimate the mean. The optimal scheme should not use more than $100$ samples. The constraints on the detection errors as well as on the estimation accuracy are summarized in the second column of \cref{tbl:simResErr}.
In order to design the optimal scheme, a sufficient statistic in the sense of \ref{ass:suffStat} has to be found. Let
\begin{align}
 \bar{x}_n &= \frac{1}{n}\sum_{k=1}^n x_k \label{eq:xbar}\\
 \bar{s}_n^2 &= \frac{1}{n} \sum_{k=1}^n \obsScalar{k}^2 - \bar{x}_n^2 \label{eq:sbar}
\end{align}
denote the sample mean and the sample variance, respectively. Then, the likelihood can be written as
\begin{align}
 p(\obs\given&\Hyp_m,\mu_m) = \nonumber \\
             & \phantom{ = }\,\,\, (2\pi\sigma^2)^{-\frac{n}{2}} \exp\biggl(-\frac{n\bar{s}_n^2 + n(\bar{x}_n - \mu_m)^2}{2\sigma^2}\biggr)\label{eq:lh}\\
                        & = (2\pi\sigma^2)^{-\frac{n}{2}} \exp\biggl(-\frac{n\bar{s}_n^2}{2\sigma^2}\biggr)\exp\biggl( - \frac{n(\bar{x}_n - \mu_m)^2}{2\sigma^2}\biggr)\,. \nonumber
\end{align}
Since the variance is known, the relation between the data and the random mean is completely described by $\bar{x}_n$, i.e.,
\begin{align*}
 p(\obs\given\Hyp_m,\mu_m) & \propto \exp\biggl( - \frac{n(\bar{x}_n - \mu_m)^2}{2\sigma^2}\biggr)
\end{align*}and hence, $\bar{x}_n$ is used as a sufficient statistic in the sense of \ref{ass:suffStat}. For the likelihood of the sufficient statistic it holds that
\begin{align*}
 p(\stat{n}\given\Hyp_m,\mu_m) \propto \exp\biggl(-\frac{n(\bar{x}_n - \mu_m)^2}{2\sigma^2}\biggr) \propto \norm{\mu_m}{\frac{\sigma^2}{n}}\,.
\end{align*}Note that the likelihood is continuous in the sufficient statistic as well as in the random parameter $\mu_m$, which itself follows a continuous distribution. Therefore, the posterior probabilities $p(\Hyp_m\given\stat{n})$ are continuous random variables with respect to $\stat{n}$ and, hence, the boundary of the stopping region is a P-null set according to \cref{lem:BoundPnull}.
The discretization of the continuous spaces is summarized in \cref{tbl:simContQuant}. 
\begin{table}[!t]
\centering
 \caption{Shift-in-Mean scenario: Simulation setup.}
\label{tbl:simContQuant}
  \begin{tabular}{@{}r c c@{}}
  \toprule
  quantity & domain & \#grid points \\
  \midrule
  $\mu$ & $[-16,16 ]$ & $7000$ \\
  $\bar{x}_n$ & $[-8,8 ]$ & $1600 $ \\
  $x$ & $[-15,15 ]$ & $6000$ \\
\bottomrule
 \end{tabular}
 \end{table}
The posterior probabilities of the hypotheses $p(\Hyp_m\given\stat{n})$ and the posterior variances $\Var\bigl[\mu_m\given\Hyp_m,\stat{n}\bigr]$ are calculated by numerical integration.
The procedure is then designed by solving the \ac{LP} in \cref{eq:LP}.
\begin{figure}[!t]
  \centering
  \includegraphics{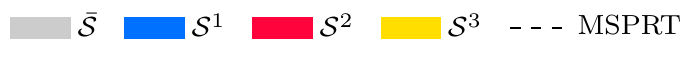}\\
\includegraphics{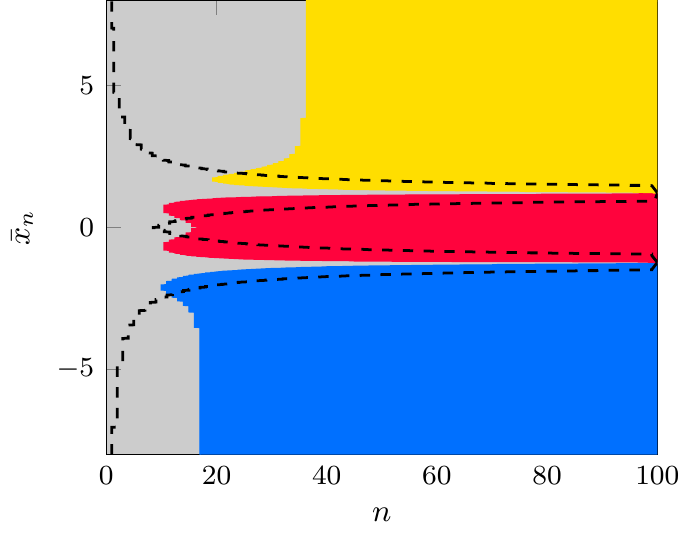}\caption{Shift-in-mean scenario: optimal and sub-optimal (\ac{MSPRT}) policy. The region in which the optimal method continues sampling is denoted by $\stopRegionCompl[]$ and the regions in which the optimal method stops and decides in favor of $\Hyp_m$ are denoted by $\stopRegionDec{}{m}$, $m\in\{1,2,3\}$. The dashed line indicates the boundaries of the \ac{MSPRT}.}\label{fig:simPolicy}
\end{figure}
The resulting coefficients are given in \cref{tbl:simCoeff} and the optimal policy parametrized by these coefficients as well as the boundaries of the \ac{MSPRT}  are shown in \cref{fig:simPolicy}. One can see that for small $n$, the optimal policy is mainly influenced by the detection constraint. For small $n$, a large value of $\bar{x}_n$ results in a high certainty of $\Hyp_3$ but, on the other hand, also in a high estimation error. Hence, the optimal scheme continues sampling in this case and stops only for $n>35$. Contrary to this, the \ac{MSPRT} stops directly for large $\bar{x}_n$, even if $n$ is small, and decides in favor of $\Hyp_3$. With increasing $n$, the optimal policy and the one of the \ac{MSPRT} become very similar. Nevertheless, the \ac{MSPRT} has a much broader corridor for continuing than the optimal procedure.

To validate the performance of the optimal and the two-step procedure, a Monte Carlo simulation with $10^6$ runs is performed. The results are summarized in \cref{tbl:simResErr} and \cref{tbl:simResRL}.
For the optimal scheme, the constraints are, within the range of Monte Carlo uncertainty, fulfilled with equality and the expected run-length obtained via \cref{eq:LP} is almost equal to the empirical average run-length. For the two-step procedure, the empirical errors under $\Hyp_2$ are close to the constraints, but the empirical average run-length of the two-step procedure under $\Hyp_2$ is larger than the one of the optimal scheme.
Moreover, under $\Hyp_1$ and $\Hyp_3$, the average run-length and the detection errors of the two-step procedure are much smaller than the ones of the optimal scheme. 
Though the two-step procedure has smaller average run-lengths and error probabilities under $\Hyp_1$ and $\Hyp_3$, the empirical \ac{MSE} is $4$ and $8$ times as large as the constraint, respectively.
\begin{table}[!t]
 \centering
  \caption{Shift-in-mean scenario: simulation results.}
 \subfloat[Detection and estimation errors.]{\label{tbl:simResErr}\setlength{\tabcolsep}{4pt}
\begin{tabular}{@{}rccc@{}}
	\toprule
	& constraints & optimal & two-step \\
	\midrule

	$\errorDet{}{1}$ & $0.050$ & $0.050$ & $0.021$ \\
	$\errorDet{}{2}$ & $0.050$ & $0.049$ & $0.053$ \\
	$\errorDet{}{3}$ & $0.050$ & $0.050$ & $0.021$ \\
	\midrule
	$\errorEst{}{1}$ & $0.200$ & $0.209$ & $0.810$ \\
	$\errorEst{}{2}$ & $0.150$ & $0.157$ & $0.182$ \\
	$\errorEst{}{3}$ & $0.100$ & $0.105$ & $0.811$ \\

	\bottomrule
\end{tabular}
 }\hfill
 \subfloat[Expected run-lengths.]{\label{tbl:simResRL}\begin{tabular}{@{}rccc@{}}
	\toprule
	& calculated & \multicolumn{2}{c}{simulated} \\
	\cmidrule(lr){2-2} \cmidrule(lr){3-4}
	& optimal & optimal & two-step \\
	\midrule
	
	$\E[\tau\given\Hyp_1]$ & - & $15.06$ & $10.13$ \\
	$\E[\tau\given\Hyp_2]$ & - & $14.472$ & $19.15$ \\
	$\E[\tau\given\Hyp_3]$ & - & $31.77$ & $10.08$ \\

	$\E[\tau]$ & $20.395$ & $20.523$ & $13.120$ \\

	\bottomrule
\end{tabular}
 }\hfill
 \subfloat[Optimal cost {coefficients}.]{\label{tbl:simCoeff}\setlength{\tabcolsep}{10pt}
\begin{tabular}{@{}r|cc@{}}
	\toprule
	$m$ & $\costDet_m^\star$ & $\costEst_m^\star$ \\
	\midrule
	$1$ & $62.99$ & $74.98$ \\
	$2$ & $82.43$ & $112.68$ \\
	$3$ & $85.69$ & $342.02$\\
	\bottomrule
\end{tabular}
 }
\end{table}

\subsection{Joint 4-ASK Decoding and Noise Power Estimation}
\acused{ASK}
This numerical example, which was partly presented in \citet{reinhard2020Multiple}, gives an example for how to apply  the proposed framework to a real-world problem.
We consider a 4-\acl{ASK} (4-\ac{ASK}) symbol to be transmitted over an additive white Gaussian noise channel with random noise power. At the receiver side, we want to jointly decode the transmitted symbol and estimate the noise power.
Using a linear model, the received signal is given by
\begin{align*}
 x_n = A + w_n\,,\end{align*}
where $A\in\{A_1,A_2,A_3,A_4\}$ denotes the \ac{ASK} symbol and $w_n$, $n=1,\ldots,N$, is the additive white Gaussian noise process.
According to assumption \ref{ass:HparamConst}, the transmitted symbol $A$ does not change during the observation period. 
The distribution of the noise power $\sigma^2$ follows an inverse Gamma distribution itself with known hyperparameters. Here, the signal decoding is a hypothesis test. Hence, the four different hypotheses can be written as
\begin{align*}
 \Hyp_m:&\; \RVidx\given\sigma^2 \overset{\text{\ac{iid}}}{\sim} \norm{A_m}{\sigma^2}\,, \sigma^2 \sim \invGam(a,b)\,, \end{align*}
where $m=1\,\ldots,4$, and $\invGam(a,b)$ is the inverse Gamma distribution with shape and scale parameters $a$ and $b$, respectively. The \ac{pdf} of the inverse Gamma distribution is given by \citep[Definition 8.22]{barber2012bayesian}
\begin{align*}
 p(\sigma^2) = \frac{b^a}{\Gamma(a)} \bigl(\sigma^2\bigr)^{-a-1}e^{-\frac{b}{\sigma^2}}\,,
\end{align*}
where $\Gamma(\cdot)$ denotes the Gamma function.
First, a sufficient statistic in the sense of \cref{ass:suffStat} has to be found. As shown in \cref{eq:lh}, the likelihood is completely determined by $\bar{x}_n$ as defined in \cref{eq:xbar} and $\bar{s}^2_n$ as defined in \cref{eq:sbar}.
Hence, the sufficient statistic $\stat{n}=[\bar{x}_n, \bar{s}_n^2]$ is used in the sequel. 
Since the inverse Gamma distribution is a conjugate prior for the variance of a Gaussian distribution, we can provide analytical expressions for all posterior quantities. First, the posterior distribution of the variance under $\Hyp_m$ follows itself an inverse Gamma distribution \citep[Section 8.8.3]{barber2012bayesian}, i.e., \begin{align*}
 \sigma^2\given\Hyp_m,\stat{n} \sim \invGam(\shapePost, \scalePost)\,,
\end{align*}with the posterior parameters
\begin{align}
  \label{eq:paramPostShape}
 \shapePost & = \shape + \frac{n}{2}\,,\\
 \begin{split}\label{eq:paramPostScale}
 \scalePost & = \scale + 0.5 \sum_{k=1}^n (x_k-A_m)^2\,\\
	    & = \scale + 0.5 n\Bigl(\bar{s}_n^2 + (\bar{x}_n - A_m)^2\Bigr)\,.\end{split}
\end{align}
The posterior mean and the posterior variance under $\Hyp_m$ are then given by \citep[Definition 8.22]{barber2012bayesian}
\begin{align*}
 \E[\sigma^2\given\Hyp_m, \stat{n}] & = \frac{\scalePost}{\shapePost-1}\,,  \quad \text{for}\; \shapePost>1\,,\\
 \Var[\sigma^2\given\Hyp_m, \stat{n}] & = \frac{(\scalePost)^2}{(\shapePost-1)^2(\shapePost-2)}\,,\quad \text{for}\; \shapePost>2\,.
\end{align*}
Moreover, an analytical expression for the posterior probabilities of the hypotheses $p(\Hyp_m\given\stat{n})$ as well as for the posterior predictive $p(\xnew\given\stat{n})$ can be expressed as
\begin{align}
 p(\Hyp_m\given\stat{n}) & =  K \bigl(2\pi\bigr)^{-\frac{n}{2}} p(\Hyp_m) \frac{\scale^\shape}{(\scalePost)^\shapePost} \frac{\Gamma(\shapePost)}{\Gamma(\shape)}\,, \label{eq:postH}\\
 p(\xnew\given\stat{n}) & = K \bigl(2\pi\bigr)^{-\frac{n+1}{2}} \sum_{m=1}^M p(\Hyp_m) \frac{\scale^\shape}{(\scalePP)^\shapePP}\frac{\Gamma(\shapePP)}{\Gamma(\shape)}  \label{eq:postPred}\,,
\end{align}where the parameters of the posterior predictive $\scalePP$, $\shapePP$ and the normalization constant $K$ are given by
\begin{align*}
 \shapePP & = \shapePost + 0.5\,, \\
 \scalePP & = \scalePost + 0.5(\xnew-A_m)^2\,, \\
 K & =\left(\sum_{m=1}^M p(\Hyp_m\given\stat{n})\right)^{-1}\,.
\end{align*}A detailed derivation of \cref{eq:postH} and \cref{eq:postPred} is laid down in \cref{app:derPostH} and \cref{app:derPostPred}, respectively. Since $p(\Hyp_m\given\stat{n})$ are continuous random variables with respect to $\stat{n}$, the boundary $\stopRegionBound$ is a P-null set according to \cref{lem:BoundPnull}.

The aim is to design an optimal sequential scheme, which uses at most $N=50$ samples while the detection and estimation errors are respectively constrained to be below $0.05$ and $0.15$ under all four hypotheses. The \ac{ASK} symbols were set to $A_m\in\{-2,-1,1,2\}$ and the parameters of the noise distribution are given by $a=2.1$ and $b=0.9$.
All hypotheses have the same prior probabilities. 

In order to design the optimal scheme, we first use the \ac{LP} approach on a coarse grid to get an initial set of cost coefficients. These cost coefficients are then used as initial values for the projected gradient ascent which we run on a finer grid.
To reduce the computational load, we exploit the symmetry of the problem at hand. This means that we solve both optimization problems only for $\costDet_m$, $\costEst_m$, $m\in\{1,2\}$ and $\rho_n(\stat{n})$, $\bar{x}_n\leq0$ and complete the missing values once the scheme is designed.

\begin{table}[!t]
 \centering
  \caption{Joint 4-ASK signal decoding and noise power estimation: simulation setup.}
  \label{tbl:4askSetup}
 \begin{tabular}{@{}c c c c c@{}}
  \toprule
  &  \multicolumn{2}{c}{coarse grid} & \multicolumn{2}{c}{fine grid} \\
  \cmidrule(lr){2-3} \cmidrule(lr){4-5}
  quantity & domain & \#grid points & domain & \#grid points \\
  \midrule
  $\bar{x}_n$ & $[-9,9]$ & $121$ &$[-14,14]$ & $243$ \\
  $\bar{s}^2_n$ & $[0,30]$ & $121$ &$[0,60]$ & $242$ \\
  $x$ &  $[-25,25]$ & $2100$ &$[-25,25]$ & $2100$ \\
  \bottomrule
 \end{tabular}
 \end{table}
As mentioned before, the \ac{LP} approach is solved on a coarse grid, whereas a finer grid is used for the gradient ascent. The discretization used for the two algorithms is summarized in \cref{tbl:4askSetup}. For the gradient ascent, the gradients are estimated by Monte Carlo simulations with $10^6$ runs. The scaling factor for the gradient is set to $\gamma=1000$ to speed up convergence.
The stopping criterion for \cref{alg:gradAsc} is
\begin{align*}
  \costDet_m = 0 \quad \lor \quad  \lvert\errorDet{}{m} - \detConstr\rvert  & \leq 10^{-3}\,,\; m=1,\ldots,M\,,\\
  \costEst_m = 0 \quad \lor \quad \lvert\errorEst{}{m} - \estConstr\rvert & \leq 5\cdot10^{-3}\,,\; m=1,\ldots,M\,.
\end{align*}
The designed scheme is evaluated using $10^6$ Monte Carlo runs.

\cref{tbl:4askRes} summarizes the Monte Carlo results for the optimal procedure, along with those of the two-step procedure. The optimal sequential scheme hits the constraints exactly, within the tolerance.
Moreover, the \ac{MSPRT} achieves smaller empirical detection errors than the constraints, but the estimation constraints are violated since the \ac{MSPRT} does not take the estimation errors into account.
In \cref{tbl:4askResRL}, the empirical run-lengths are summarized for both procedures. Though the two-step procedure has a smaller empirical run-length than the optimal one, this comes at the cost of violating the constraints on the \ac{MSE} as mentioned previously.

\begin{table}[!t]
 \centering
  \caption{Joint 4-ASK signal decoding and noise power estimation: simulation results.}
  \label{tbl:4askRes}
 \subfloat[Detection and estimation errors.\label{tbl:4askResErr}]{\setlength{\tabcolsep}{4pt}
\begin{tabular}{@{}rcccc@{}}
	& \multicolumn{2}{c}{constraints} & \multicolumn{2}{c}{empirical}\\	
	\cmidrule(lr){2-3} \cmidrule(lr){4-5}
	& $\detConstr[]$/$\estConstr[]$&tolerance & optimal & two-step \\
	\midrule
	$\errorDet{}{1}$ & $0.050$ & $\pm0.001$ &$0.051$ & $0.029$ \\
	$\errorDet{}{2}$ & $0.050$ & $\pm0.001$ &$0.049$ & $0.038$ \\
	$\errorDet{}{3}$ & $0.050$ & $\pm0.001$ &$0.049$ & $0.037$ \\
	$\errorDet{}{4}$ & $0.050$ & $\pm0.001$ &$0.051$ & $0.030$ \\
	\midrule
	$\errorEst{}{1}$ & $0.150$ & $\pm0.005$ &$0.151$ & $0.320$ \\
	$\errorEst{}{2}$ & $0.150$ & $\pm0.005$ &$0.151$ & $0.246$ \\
	$\errorEst{}{3}$ & $0.150$ & $\pm0.005$ &$0.153$ & $0.243$ \\
	$\errorEst{}{4}$ & $0.150$ & $\pm0.005$ &$0.149$ & $0.311$ \\
	\bottomrule
\end{tabular}
 }\hspace{1cm}
 \subfloat[Expected run-lengths.\label{tbl:4askResRL}]{\setlength{\tabcolsep}{2pt}
\begin{tabular}{@{}rcc@{}}
	 & optimal & two-step \\
	\midrule
	$\E[\tau\given\Hyp_1]$ & $6.54$&$5.67$\\
	$\E[\tau\given\Hyp_2]$ & $6.33$&$5.93$\\
	$\E[\tau\given\Hyp_3]$ & $6.32$&$5.92$\\
	$\E[\tau\given\Hyp_4]$ & $6.54$&$5.67$\\
	\midrule
	$\E[\tau]$ & $6.43$&$5.80$\\
	\bottomrule
\end{tabular}
 }
\end{table}
\begin{figure*}[!t]
 \centering
 \includegraphics{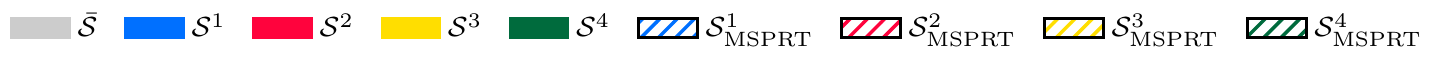}\\
 \vspace{-1.5em}
 \subfloat[$n=5$\label{fig:4askPolicy_n5}]{\includegraphics{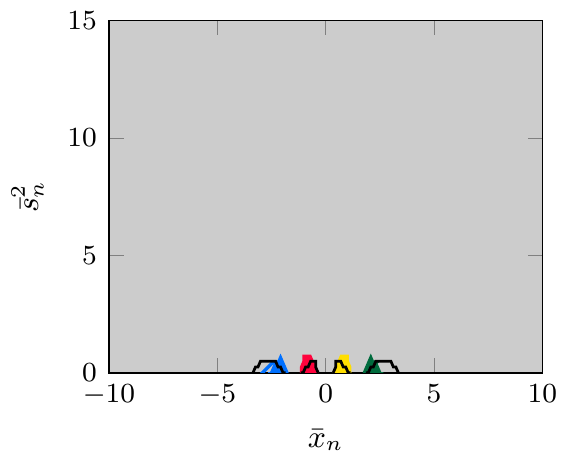}}\hfil
 \subfloat[$n=10$\label{fig:4askPolicy_n10}]{\includegraphics{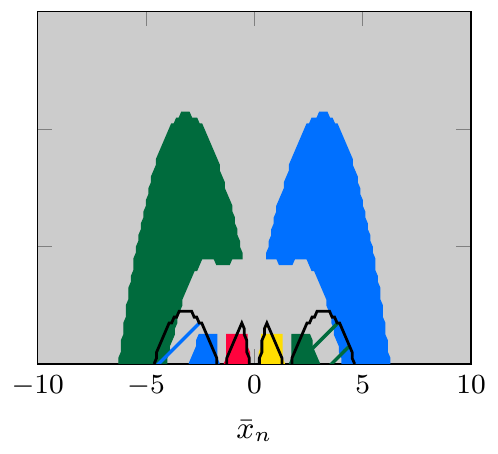}}\hfil
 \subfloat[$n=30$\label{fig:4askPolicy_n30}]{\includegraphics{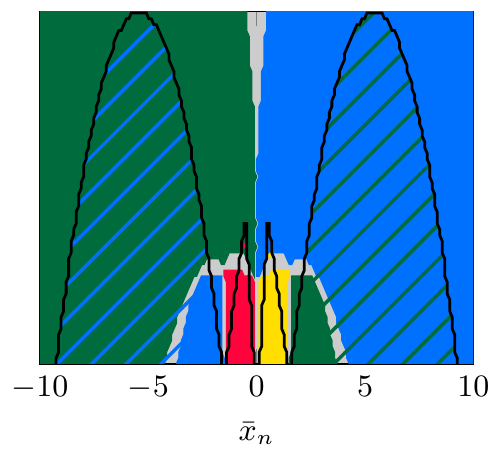}}\hfil
 \caption{Joint 4-ASK signal decoding and noise power estimation: evolution of the optimal policy (filled) and \ac{MSPRT} policy (hatched) over time. The region in which the optimal procedure continues sampling is denoted by $\stopRegionCompl[]$ and the regions in which the optimal procedure stops and decides in favor of $\Hyp_m$ are denoted by $\stopRegionDec{}{m}$. The regions in which the \ac{MSPRT} stops and decides in favor of $\Hyp_m$ are denoted by $\stopRegionDecMSPRT{m}$, i.e., $\stopRegionDecMSPRT{m}=\bigl\{\stat{n}\in\stateSpaceStat: \stopR_n^\text{MSPRT}(\stat{n})=1 \wedge \dec_n^\text{MSPRT}(\stat{n})=m\bigr\}$.}
 \label{fig:4askPolicyEvolution}
\end{figure*}
In \cref{fig:4askPolicyEvolution}, the evolution of both policies over time is shown for three distinct time instances. In this figure, the gray region is the complement of the stopping region of the optimal scheme and the other filled regions are the regions in which the optimal scheme stops and decides for a particular hypothesis.
The hatched areas indicate the regions in which the \ac{MSPRT} stops and decides in favor of a particular hypothesis.
For $n=5$ (\cref{fig:4askPolicy_n5}), most of the state space corresponds to the complement of the stopping region of the methods and there are four small areas in which both procedures, the optimal and the two-step, stop.
For both methods, these regions are still present for $n=10$, but larger compared to the ones in \cref{fig:4askPolicy_n5}.
In addition, there appear two regions for small/large values of $\bar{x}_n$ in which the optimal procedure stops and decides in favor of $\Hyp_4$/$\Hyp_1$.
Since a small value of $\bar{x}_n$ increases the certainty about $\Hyp_1$, a decision in favor of $\Hyp_4$ is not intuitive here.
As we consider a joint detection and estimation problem, the uncertainty about the true hypothesis and the true parameter affect the decision rule. Though a small value of $\bar{x}_n$ implies a high certainty about $\Hyp_1$, it leads to a high uncertainty about the true parameter at the same time. Hence, a decision in favor of $\Hyp_4$ is less costly.
For the \ac{MSPRT}, which does not encounter any estimation, this phenomenon is not visible at all.
The four regions for stopping the optimal procedure, which are present at the bottom of \cref{fig:4askPolicy_n5} and \cref{fig:4askPolicy_n10}, have grown further in \cref{fig:4askPolicy_n30}. Though these four regions have grown equally for the optimal method, the regions for stopping the \ac{MSPRT} and deciding in favor of $\Hyp_1$ and $\Hyp_4$ now cover a large area of the state space. At the same time, the two regions of the optimal policy, that appeared in \cref{fig:4askPolicy_n10}, cover almost the whole state space but make a decision opposite to the \ac{MSPRT}.

  \section{Conclusions}
Based on very mild assumptions on the underlying stochastic process, we have developed a flexible framework for sequential joint multiple hypotheses testing and parameter estimation.
The optimal scheme minimizes the expected run-length while fulfilling constraints on the probabilities of falsely rejecting a hypothesis as well as on the \ac{MSE}
. These procedures have been characterized by a set of non-linear Bellman equations. We have further shown a strong connection of the cost coefficients of the Bellman equations and the performance measures of the method, i.e., the probability of falsely rejecting a hypothesis and the \acp{MSE}. Based on this connection, we have presented two approaches to obtain the set of optimal cost coefficients. The first approach formulates the problem as a linear program and the second approach uses a projected gradient ascent.
The performance of the optimal procedure has been validated via two numerical examples. The first example has shown the basic properties of the optimal scheme, whereas the second example has been used to show applicability to real world problems. For both examples, Monte Carlo results have been provided. Moreover, the performance gap to a sub-optimal scheme, i.e., a matrix sequential probability ratio test followed by an \ac{MMSE} estimator, has been shown.
 
\appendix
\crefalias{section}{appsec}

\section{Proof of Corollary \ref{corr:integrable}}\label{sec:proofIntegrable}
It has to be shown that if all $\costDet_m$ and all $\costEst_m$, $m=1,\ldots,M$, are finite, then $\rho_{n+1}$ is $\updateProbMeasure$-integrable for all $\stat{n}\in\stateSpaceStat$ and all $0\leq n < N$.
From the definition of the cost function $\rho_n$, one can directly see that
\begin{align*}
 \int \rho_{n+1}\dInt\updateProbMeasure \leq  \int g\dInt\updateProbMeasure \leq \int D^\star_{m,n+1}\dInt\updateProbMeasure\,.
\end{align*}
With the definition of $D^\star_{n+1,m}$, $m=1,\ldots,M$, the integral on the right hand side can be written as
\begin{align}
 \int D^\star_{m,n+1}\dInt\updateProbMeasure & = \int \costEst_m p(\Hyp_m\given\stat{n+1}) \Var[\paramRV_m\given\Hyp_m, \stat{n+1}] + \sum_{i=1,i\neq m}^M \costDet_i p(\Hyp_i\given\stat{n+1}) \dInt\updateProbMeasure \nonumber \\
 &   = \costEst_m \int p(\Hyp_m\given\stat{n}, \obsScalar{n+1}) \Var[\paramRV_m\given\Hyp_m, \stat{n}, \obsScalar{n+1}]p(\obsScalar{n+1}\given\stat{n})\dInt\obsScalar{n+1} \label{eq:expMSE}\\
  & + \sum_{i=1,i\neq m}^M \costDet_i \int  p(\Hyp_i\given\stat{n}, \obsScalar{n+1}) p(\obsScalar{n+1}\given\stat{n})\dInt\obsScalar{n+1} \label{eq:expDet}\,.    
\end{align}
The integral in \cref{eq:expMSE} can be written as
\begin{align}
 \nonumber & p(\Hyp_m\given\stat{n})\int \Var[\paramRV_m\given\Hyp_m,\stat{n},\obsScalar{n+1}]p(\obsScalar{n+1}\given\Hyp_m,\stat{n})\dInt\obsScalar{n+1}\,\\
 ={}&p(\Hyp_m\given\stat{n})\int\int \Bigl(\param_m - \E[\paramRV_m\given\Hyp_m,\stat{n},\obsScalar{n+1}]\Bigr)^2p(\param_m,\obsScalar{n+1}\given\Hyp_m,\stat{n})\dInt\obsScalar{n+1}\dInt\param_m\,.\label{eq:post_var_Q-integrable_double}
\end{align}
By expanding the square, the double integral in \cref{eq:post_var_Q-integrable_double} becomes
\begin{align*}
& \int\int \param_m^2 p(\param_m,\obsScalar{n+1}\given\Hyp_m,\stat{n})\dInt\obsScalar{n+1}\dInt\param_m \\
 & - 2 \int\int \param_m\E[\paramRV_m\given\Hyp_m,\stat{n},\obsScalar{n+1}] p(\param_m,\obsScalar{n+1}\given\Hyp_m,\stat{n})\dInt\obsScalar{n+1}\dInt\param_m \\
 & + \int\int \Bigl(\E[\paramRV_m\given\Hyp_m,\stat{n},\obsScalar{n+1}]\Bigr)^2 p(\param_m,\obsScalar{n+1}\given\Hyp_m,\stat{n})\dInt\obsScalar{n+1}\dInt\param_m\,.\\
={}& \E[ \paramRV_m^2 \given\Hyp_m,\stat{n}]  - 2 \Bigl(\E[\paramRV_m\given\Hyp_m,\stat{n}]\Bigr)^2 \\
 &  + \int\int \Bigl(\E[\paramRV_m\given\Hyp_m,\stat{n},\obsScalar{n+1}]\Bigr)^2 p(\param_m,\obsScalar{n+1}\given\Hyp_m,\stat{n})\dInt\obsScalar{n+1}\dInt\param_m\,.
\end{align*}
By applying Jensen's inequality \citep[p. 228]{everitt2010cambridge}, the integral in the last equation can be upper bounded by
\begin{align*}
 & \int\int \Bigl(\E[\paramRV_m\given\Hyp_m,\stat{n},\obsScalar{n+1}]\Bigr)^2 p(\param_m,\obsScalar{n+1}\given\Hyp_m,\stat{n})\dInt\obsScalar{n+1}\dInt\param_m \\
 \leq{} &  \int\int \E[\paramRV_m^2\given\Hyp_m,\stat{n},\obsScalar{n+1}] p(\param_m,\obsScalar{n+1}\given\Hyp_m,\stat{n})\dInt\obsScalar{n+1}\dInt\param_m  \\
 ={}& \E[\paramRV_m^2\given\Hyp_m,\stat{n}]\,.
\end{align*}
Hence, by combining the previous results, it follows for the integral in  \cref{eq:expMSE} that
\begin{align*}
 \int p(\Hyp_m\given\stat{n} \obsScalar{n+1}) \Var[\paramRV_m\given\Hyp_m,\stat{n}, \obsScalar{n+1}]p(\obsScalar{n+1}\given\stat{n})\dInt\obsScalar{n+1} \leq 2p(\Hyp_m\given\stat{n}) \Var[\paramRV_m\given\Hyp_m,\stat{n}]\,.
\end{align*}
Moreover, \cref{eq:expDet} reduces to
\begin{align*}
 \sum_{i=1,i\neq m}^M \costDet_i \int  p(\Hyp_i\given\stat{n}, \obsScalar{n+1}) p(\obsScalar{n+1}\given\stat{n})\dInt\obsScalar{n+1} = \sum_{i=1,i\neq m}^M \costDet_i p(\Hyp_i\given\stat{n})\,.
\end{align*}
Hence, we can conclude that
\begin{align*}
 \int \rho_{n+1}\dInt\updateProbMeasure & \leq  \int D^\star_{m,n+1}\dInt\updateProbMeasure \\
 & \leq  \costEst_m 2p(\Hyp_m\given\stat{n})\Var[\paramRV_m\given\Hyp_m, \stat{n}] + \sum_{i=1,i\neq m}^M \costDet_i p(\Hyp_i\given\stat{n}) < \infty\,,
\end{align*}
which is finite as long as all $\costDet_m$, $\costEst_m$, $m=1,\ldots,M$, are finite since the posterior probabilities of $\Hyp_m$ are finite by definition and the posterior variance is finite by assumption \cref{ass:finiteSOM}. \hfill\qed

\section{Proof of Lemma \ref{lem:BoundPnull}}\label{sec:proofBoundPnull}
In order to prove \cref{lem:BoundPnull}, we first transfer the cost functions defined in \cref{theo:optimalStoppingSol}, the boundary of the stopping region defined in \cref{eq:regions} and the probability measure defined in \cref{eq:updateProbMeasure} to another domain. In the new domain they depend on the sufficient statistic as well as on the posterior probabilities. The posterior probability of $\Hyp_m$, $m=1,\ldots,M$, is denoted by $\postProb{m}$ in what follows. The posterior probabilities are collected in the tuple $\postProbTuple=(\postProb{1},\ldots,\postProb{M})$, which is defined on the metric state space $\left(\stateSpacePostProb, \metricPostProb\right)$.
Hence, the combined cost for deciding in favor of $\Hyp_m$ is given by
\begin{align*}
 \auxVarCostOptTilde{m,n}(\stat{n},\postProbTuple) = \sum_{i=1,i\neq m}^M \costDet_i \postProb{i} + \costEst_m \postProb{m} \Var[\paramRV_m\given\Hyp_m, \stat{n}]\,.
\end{align*}
We can now rewrite the overall cost function as
\begin{align*}
 \tilde\rho_n(\stat{n},\postProbTuple) & = \min\left\{\tilde g(\stat{n}, \postProbTuple), \tilde d_n(\stat{n}, \postProbTuple) \right\}\,\quad n<N\,,\\
 \tilde\rho_N(\stat{N},\postProbTuple[N]) & = \tilde g(\stat{N}, \postProbTuple[N])\,,
\end{align*}
where the cost functions for continuing and stopping the test are defined as
\begin{align*}
 \tilde g(\stat{n}, \postProbTuple) &=  \min\left\{\auxVarCostOptTilde{1,n}(\stat{n},\postProbTuple), \ldots, \auxVarCostOptTilde{M,n}(\stat{n},\postProbTuple)\right\}\,, \\
 \tilde d_n(\stat{n}, \postProbTuple) &=  1 + \int \tilde\rho_{n+1}(\transkernel(\xnew), \transkernelPostVar{\postProbTuple}{\xnew}{\stat{n}}) p(\xnew\given\stat{n})\dInt \xnew\,.
\end{align*}
The transition kernel of the posterior probabilities is given by $\postProbTuple[n+1] = \transkernelPostVar{\postProbTuple}{\xnew}{\stat{n}}$. The equivalent of the probability measure defined in \cref{eq:updateProbMeasure} is given by
\begin{align*}
 \updateProbMeasureTilde(B\times\tilde B) = P\left(\left\{\xnew\in\stateSpaceObs: \transkernel(\xnew)\in B, \transkernelPostVar{\postProbwo_n}{\xnew}{\stat{n}} \in\tilde B \right\}\right)\,,
\end{align*}
for all elements $B$ of the Borel $\sigma$-algebra on $\stateSpaceStat$ and all elements $\tilde B$ of the $\sigma$-algebra on $\stateSpacePostProb$. Finally, the counterpart of the boundary of the stopping region defined in \cref{eq:regions} is given by
\begin{align*}
 \stopRegionBoundTilde = \left\{(\stat{n}, \postProbTuple)\in\stateSpaceStat\times\stateSpacePostProb: \tilde g(\stat{n}, \postProbTuple) = \tilde d_n(\stat{n}, \postProbTuple) \right\}\,.
\end{align*}
Before we can prove that the boundary of the stopping region is a P-null set, two auxiliary lemmas have to be introduced.
\begin{lemma} \label{lem:gIneq}
 Let $\textbf{a}=(a_1,\ldots,a_M)$ and let $\textbf{a}\cdot\postProbTuple$ denote the element-wise product. Then for all $a\in\nonNegSet^M$, all $\stat{n}\in\stateSpaceStat$ and all $\postProbTuple\in\stateSpacePostProb$ it holds that
 \begin{align*}
  \min\{a_1,\ldots,a_M,1\}\tilde g(\stat{n},\postProbTuple) \leq \tilde g(\stat{n},\mathbf{a}\cdot\postProbTuple) \leq \max\{a_1,\ldots,a_M,1\}\tilde g(\stat{n},\postProbTuple)\,.
 \end{align*}
\end{lemma}

\begin{proof}
 Since the proof for the upper and lower bound do not differ significantly, only the proof for the lower bound is outlined here. Let $a^\star=\min\{a_1,\ldots,a_M\}$, then it holds that
 \begin{align}\label{eq:gTildeAstar}
  \tilde g(\stat{n}, \mathbf{a}\cdot\postProbTuple) & = \min\left\{\auxVarCostOptTilde{1,n}(\stat{n},\mathbf{a}\cdot\postProbTuple), \ldots, \auxVarCostOptTilde{M,n}(\stat{n},\mathbf{a}\cdot\postProbTuple)\right\} \\
 & = a^\star\min\left\{\frac{1}{a^\star}\auxVarCostOptTilde{1,n}(\stat{n},\mathbf{a}\cdot\postProbTuple), \ldots, \frac{1}{a^\star}\auxVarCostOptTilde{M,n}(\stat{n},\mathbf{a}\cdot\postProbTuple)\right\}\,.
 \end{align}
It further holds that
\begin{align} \label{eq:astarIneq}
 \frac{1}{a^\star}\auxVarCostOptTilde{m,n}(\stat{n},\mathbf{a}\cdot\postProbTuple) = \sum_{i=1,i\neq m}^M \frac{a_i}{a^\star}\costDet_i \postProb{i} + \costEst_m \frac{a_m}{a^\star}\postProb{m} \Var[\paramRV_m\given\Hyp_m, \stat{n}] \geq \auxVarCostOptTilde{m,n}(\stat{n},\postProbTuple)\,,
\end{align}
since $a^\star \leq a_m,\;\forall m\in\{1,\ldots,M\}$.

 Applying \cref{eq:astarIneq} to \cref{eq:gTildeAstar}, yields
 \begin{align*}
  \tilde g(\stat{n}, \mathbf{a}\cdot\postProbTuple) & = \min\left\{\auxVarCostOptTilde{1,n}(\stat{n},\mathbf{a}\cdot\postProbTuple), \ldots, \auxVarCostOptTilde{M,n}(\stat{n},\mathbf{a}\cdot\postProbTuple)\right\} \\
 & = a^\star\min\left\{\frac{1}{a^\star}\auxVarCostOptTilde{1,n}(\stat{n},\postProbTuple), \ldots, \frac{1}{a^\star}\auxVarCostOptTilde{M,n}(\stat{n},\postProbTuple)\right\} \\
 & \geq a^\star\min\left\{\auxVarCostOptTilde{1,n}(\stat{n},\postProbTuple), \ldots, \auxVarCostOptTilde{M,n}(\stat{n},\postProbTuple)\right\} \\
 & \geq \min\{a^\star,1\}\min\left\{\auxVarCostOptTilde{1,n}(\stat{n},\postProbTuple), \ldots, \auxVarCostOptTilde{M,n}(\stat{n},\postProbTuple)\right\} = \min\{a_1,\ldots,a_M,1\}\tilde g(\stat{n}, \postProbTuple)\,,
 \end{align*}
 which is the lower bound stated in \cref{lem:gIneq}.
\end{proof}

\begin{lemma} \label{lem:rhoIneq}
 Let $\mathbf{a}=(a_1,\ldots,a_M)$ and let $\mathbf{a}\cdot\postProbTuple$ denote the element-wise product. Then for all $\mathbf{a}\in\nonNegSet^M$, all $\stat{n}\in\stateSpaceStat$ and all $\postProbTuple\in\stateSpacePostProb$ it holds that
 \begin{align*}
  \min\{a_1,\ldots,a_M,1\}\tilde \rho_n(\stat{n},\postProbTuple) \leq \tilde \rho_n(\stat{n},\mathbf{a}\cdot\postProbTuple) \leq \max\{a_1,\ldots,a_M,1\}\tilde \rho_n(\stat{n},\postProbTuple)\,.
 \end{align*}
\end{lemma}
\begin{proof}
 Since the proofs for the upper and lower bound do not differ significantly, only the proof for the lower bound is outlined here. First of all, it has to be mentioned that the transition kernel relating the posterior probabilities $\postProb{m}$ and $\postProb[n+1]{m}$ is linear in $\postProb{m}$, i.e.,
 \begin{align}\label{eq:transKernIneq}
  \postProb[n+1]{m} = \transkernelPostVar{\postProbwo_n}{\obsScalar{n+1}}{\stat{n}} = \postProbwo_n \frac{p(\obsScalar{n+1}\given\Hyp_m,\stat{n})}{p(\obsScalar{n+1}\given\stat{n})}\,.
 \end{align}
The proof is done via induction. Let $a^\star = \min\{a_1,\ldots,a_M\}$ and assume that \cref{lem:rhoIneq} holds for some $0<n<N$. Then, by applying \cref{eq:transKernIneq}, it holds for $n-1$ that
\begin{align*}
 \tilde \rho_{n-1}(\stat{n-1},\mathbf{a}\cdot\postProbTuple[n-1]) & = \min\left\{ \tilde g(\stat{n-1},\mathbf{a}\cdot\postProbTuple[n-1]), 1 + \int \tilde \rho_{n}(\transkernel[n-1](\obsScalar{n}), \transkernelPostVar{\mathbf{a}\cdot\postProbTuple[n-1]}{\obsScalar{n}}{\stat{n-1}} p(\obsScalar{n}\given\stat{n-1}) \dInt \obsScalar{n} \right\} \\
 & = \min\left\{ \tilde g(\stat{n-1},\mathbf{a}\cdot\postProbTuple[n-1]), 1 + \int \tilde \rho_{n}(\transkernel[n-1](\obsScalar{n}), \mathbf{a}\cdot\transkernelPostVar{\postProbTuple[n-1]}{\obsScalar{n}}{\stat{n-1}} p(\obsScalar{n}\given\stat{n-1}) \dInt \obsScalar{n} \right\}\,.
\end{align*}
By applying \cref{lem:rhoIneq} and \cref{lem:gIneq}, one obtains
\begin{align*}
 \tilde \rho_{n-1}(\stat{n-1},\mathbf{a}\cdot\postProbTuple[n-1]) & \geq  \min\left\{ \tilde g(\stat{n-1},\mathbf{a}\cdot\postProbTuple[n-1]), 1 + a^\star \int \tilde \rho_{n}(\transkernel[n-1](\obsScalar{n}),\transkernelPostVar{\postProbTuple[n-1]}{\obsScalar{n}}{\stat{n-1}}) p(\obsScalar{n}\given\stat{n-1}) \dInt \obsScalar{n} \right\} \\
 & \geq  \min\left\{ a^\star \tilde g(\stat{n-1}, \postProbTuple[n-1]), 1 + a^\star \int \tilde \rho_{n}(\transkernel[n-1](\obsScalar{n}),\transkernelPostVar{\postProbTuple[n-1]}{\obsScalar{n}}{\stat{n-1}}) p(\obsScalar{n}\given\stat{n-1}) \dInt \obsScalar{n} \right\} \\
 & = \min\left\{ a^\star \tilde g(\stat{n-1}, \postProbTuple[n-1]), 1 + a^\star \int \tilde \rho_{n} \dInt \updateProbMeasureTilde[n-1] \right\} \,.
\end{align*}
We can further state that
\begin{align*}
  \min\left\{ a^\star \tilde g(\stat{n-1}, \postProbTuple[n-1]), 1 + a^\star \int \tilde \rho_{n} \dInt \updateProbMeasureTilde[n-1] \right\} & \geq 
   \min\left\{ a^\star \tilde g(\stat{n-1}, \postProbTuple[n-1]), a^\star + a^\star \int \tilde \rho_{n} \dInt \updateProbMeasureTilde[n-1] \right\} \\
   & \geq a^\star \min\left\{ \tilde g(\stat{n-1}, \postProbTuple[n-1]), 1 +  \int \tilde \rho_{n} \dInt \updateProbMeasureTilde[n-1] \right\}\,.
\end{align*}
With these results, we can conclude that \cref{lem:rhoIneq} holds for $n-1$ if it holds for $n$. The induction basis is given by $n=N$, where it holds that
\begin{align*}
 \tilde\rho_N(\stat{N},\mathbf{a}\cdot \postProbTuple[N]) = \tilde g_N(\stat{N},\mathbf{a}\cdot \postProbTuple[N]) \geq a^\star \tilde g_N(\stat{N}, \postProbTuple[N]) = a^\star \rho_N(\stat{N},  \postProbTuple[N])\,.
\end{align*}
\end{proof}

Now, with the help of \cref{lem:rhoIneq} and \cref{lem:gIneq}, \cref{lem:BoundPnull} can be proven easily by contradiction. Assume, that there exists a non-zero probability $\updateProbMeasureTilde(\stopRegionBoundTilde)$ that the test hits the boundary of the stopping region with its next update for some $n<N$. Since the posterior probabilities $\postProbTuple[N]$ are assumed to be continuous random variables a $\stat{n}\in\stateSpaceStat$ and an interval $[\postProbTuple^\bullet, \mathbf{a}\cdot\postProbTuple^\bullet]$ with $\mathbf{a}=(a_1,\ldots,a_M)$ and $a_m>1$ for all $m=1,\ldots,M$ have to exist for which the costs for stopping and continuing are equal. Mathematically, this can be written as
\begin{align}\label{eq:eqCostAssumption}
 \tilde g(\stat{n},\mathbf{a}\cdot \postProbTuple[n]) = 1 + \int \tilde\rho_{n+1} \dInt\updateProbMeasureTilde \quad \forall \postProbwo_n \in [\postProbTuple^\bullet, \mathbf{a}\cdot\postProbTuple^\bullet]\,.
\end{align}
With the previous results we can conclude that
\begin{align*}
 1 + \int \tilde\rho_{n+1} \dInt\updateProbMeasureTildeScaled{\mathbf{a}\cdot} & = 1 + \int \tilde \rho_{n+1}(\transkernel(\obsScalar{n+1}), \transkernelPostVar{\mathbf{a}\cdot\postProbTuple}{\obsScalar{n+1}}{\stat{n}}) p(\obsScalar{n+1}\given\stat{n}) \dInt \obsScalar{n+1} \\
 & = 1 + \int \tilde \rho_{n+1}(\transkernel(\obsScalar{n+1}), \mathbf{a} \cdot\transkernelPostVar{\postProbTuple}{\obsScalar{n+1}}{\stat{n}}) p(\obsScalar{n+1}\given\stat{n}) \dInt \obsScalar{n+1} \\ 
 & \geq 1 + a^\star \int \tilde \rho_{n+1}(\transkernel(\obsScalar{n+1}), \transkernelPostVar{\postProbTuple}{\obsScalar{n+1}}{\stat{n}}) p(\obsScalar{n+1}\given\stat{n}) \dInt \obsScalar{n+1} \\
 & > a^\star + a^\star\int \tilde \rho_{n+1}(\transkernel(\obsScalar{n+1}), \transkernelPostVar{\postProbTuple}{\obsScalar{n+1}}{\stat{n}}) p(\obsScalar{n+1}\given\stat{n}) \dInt \obsScalar{n+1} \\
 & = a^\star\left(1 + \int \tilde\rho_{n+1} \dInt\updateProbMeasureTildeScaled{\mathbf{a}\cdot}\right)
 = a^\star g(\stat{n},\postProbTuple) \geq g(\mathbf{a}\cdot\postProbTuple)\,,
\end{align*}
where $a^\star=\min\{a_1,\ldots,a_M\}$. The first and last inequality are to \cref{lem:rhoIneq} and \cref{lem:gIneq}, respectively. This contradicts the assumption that there exist a $\stat{n}\in\stateSpaceStat$ and an interval $[\postProbTuple^\bullet, \mathbf{a}\cdot\postProbTuple^\bullet]$ in which the costs for stopping and continuing the test are equal.
Due to the fact that $\stopRegionBoundTilde$ and $\stopRegionBound$ only differ in their representation, this contradiction also implies that $\updateProbMeasure(\stopRegionBound)=0, \; \forall \stat{n}\in\stateSpaceStat$. This concludes the proof. \hfill\qed

\section{Proof of Lemma \ref{lem:derivatives}}\label{app:proofTheoDerivatives}
Let
\begin{align*}
 \rho^\prime_{n,\costDet_m}(\stat{n}) & = \frac{\partial\rho_n(\stat{n})}{\partial \costDet_m} = \begin{cases}
      g_{\costDet_m}^\prime(\stat{n}) & \text{for } \stat{n}\in\stopRegion \\
       \frac{\partial}{\partial \costDet_m} \int\rho_{n+1}\dInt\updateProbMeasure& \text{for } \stat{n}\in\stopRegionCompl
     \end{cases}
\end{align*}denote the derivative of $\rho_n$ with respect to $\costDet_m$, which is defined everywhere on $\stateSpaceStat\setminus\stopRegionBound$.
Assume for now that the order of differentiation and integration can be interchanged, i.e.,
\begin{align*}
 \frac{\partial}{\partial \costDet_m} \int\rho_{n+1}\dInt\updateProbMeasure = \int\rho_{n+1,\costDet_m}^\prime\dInt\updateProbMeasure\,.
\end{align*}
In general, the derivative of $\rho_n$ with respect to $\costDet_m$ can now be written as
\begin{align*}
 \rho^\prime_{n,\costDet_m}(\stat{n}) = \ind{\stopRegion}g_{\costDet_m}^\prime(\stat{n}) + \ind{\stopRegionCompl}\int\rho_{n+1,\costDet_m}^\prime\dInt\updateProbMeasure\,.
\end{align*}
On the stopping region $\stopRegion$, it holds that
\begin{align*}
 g_{\costDet_m}^\prime(\stat{n}) & = \sum_{i=1,i\neq m}^M\ind{\stopRegionDec{n}{i}} p(\Hyp_m\given\stat{n})\,.
\end{align*}
With the short-hand notation
\begin{align*}
 z_n^m = \frac{p(\stat{n}\given\Hyp_m)}{p(\stat{n})}
\end{align*}
and the property $p(\stat{n}\given\stat{n-1})=p(\obsScalar{n}\given\stat{n-1})$, we can write
\begin{align}
 g_{\costDet_m}^\prime(\stat{n}) & = \sum_{i=1,i\neq m}^M\ind{\stopRegionDec{n}{i}} p(\Hyp_m\given\stat{n}) = \sum_{i=1,i\neq m}^M\!\!\!\ind{\stopRegionDec{n}{i}} p(\Hyp_m)z_{n-1}^m\frac{p(\obsScalar{n}\given\Hyp_m,\stat{n-1})}{p(\obsScalar{n}\given\stat{n-1})}\,. \label{eq:gPrimeCostDet}
\end{align}
With the use of \cref{eq:gPrimeCostDet}, we can further state that
\begin{align*}
 \int\rho_{n+1,\costDet_m}^\prime\dInt\updateProbMeasure &  = \sum_{i=0,i \neq m}^{M}\int\displaylimits_{\{\transkernel\in\stopRegionDec{n+1}{i}\}}\!\!\! p(\Hyp_m)z_{n}^mp(\xnew\given\Hyp_m,\stat{n})\dInt\xnew  + \int_{\{\stopRegionCompl[n+1]\}} \rho^\prime_{n+1,\costDet_m}\dInt\updateProbMeasure \\
 & = \sum_{i=0,i \neq m}^{M}p(\Hyp_m)z_{n}^m\updateProbMeasure^m\bigl(\stopRegionDec{n+1}{i}\bigr) + \int_{\{\stopRegionCompl[n+1]\}} \rho^\prime_{n+1,\costDet_m}\dInt\updateProbMeasure\\
  & = p(\Hyp_m)z_{n}^m\updateProbMeasure^m\bigl(\stopRegionDec{n+1}{\bar{m}}\bigr)  + \int_{\{\stopRegionCompl[n+1]\}} \rho^\prime_{n+1,\costDet_m}\dInt\updateProbMeasure\,.
\end{align*}
Hence, the derivative of $\rho_n$ with respect to $\costDet_m$ is
\begin{align*}
 \rho_{n,\costDet_m}^\prime(\stat{n}) & = \ind{\stopRegionDec{n}{\bar m}} p(\Hyp_m) z^m_n  + \ind{\stopRegionCompl}\biggl( p(\Hyp_m)z_{n}^m\updateProbMeasure^m\bigl(\stopRegionDec{n+1}{\bar{m}}\bigr) + \int_{\{\stopRegionCompl[n+1]\}} \rho^\prime_{n+1,\costDet_m}\dInt\updateProbMeasure      \biggr)\,,
\end{align*}
which is the expression stated in \cref{lem:derivatives}.
It is left to show that the order of the integration and the differentiation with respect to $\costDet_m$ can be interchanged.
According to the differentiation lemma \citep[Lemma 16.2.]{bauer2001measure} this holds if and only if
\begin{enumerate}
 \item The function $\rho_{n+1}(\stat{n+1})$ has to be $\updateProbMeasure$-integrable for all $0\leq n < N$ and all $\stat{n}\in\stateSpaceStat$.
 \item The function $\rho_n(\stat{n})$ has to be differentiable for all $0\leq n \leq N$ and all $\stat{n}\in\stateSpaceStat$.
 \item A set of functions $h_n^{m}$ and $f_n^m$ has to exist, which are independent of $\costDet_m$ and $\costEst_m$, respectively. It must further hold
 \begin{align}
  \lvert \rho^\prime_{n,\costDet_m}(\stat{m}) \rvert & \leq h_n^m(\stat{n})\,,\quad \forall m\in\{1,\ldots,M\}\,, \\
  \lvert \rho^\prime_{n,\costEst_m}(\stat{m}) \rvert & \leq f_n^m(\stat{n})\,,\quad \forall m\in\{1,\ldots,M\}\,.
 \end{align}
\end{enumerate}
Condition 1 is true by \cref{corr:integrable}, but conditions 2 and 3 have still to be proven. The proof is only carried out for the derivatives with respect to $\costDet_m$, since it can be proven analogously with the derivatives with respect to $\costEst_m$.
Assume that the differentiation lemma holds for some $n\geq 1$ and some $m\in\{1,\ldots,M\}$. It has now to be shown that the differentiation lemma holds for $n-1$ as well. On the stopping region, the derivative is given by
\begin{align*}
 \rho_{n-1,\costDet_m}^\prime(\stat{m}) = \ind{\stopRegionDec{n}{\bar m}} p(\Hyp_m\given\stat{n})\,,
\end{align*}
which is well defined and bounded.
On the complement of the stopping region, it holds that
\begin{align}
 \frac{\partial }{\partial \costDet_m} \rho_{n-1}(\stat{n})& = \frac{\partial}{\partial \costDet_m} \int \rho_{n} \dInt \updateProbMeasure[n-1] = 
 p(\Hyp_m)z_n^m\updateProbMeasure[n-1]^m(\stopRegionDec{n}{\bar m}) + \int_{\stopRegionCompl[n]} \rho^\prime_{n,\costDet_m} \dInt\updateProbMeasure[n-1]\,.
\end{align}
Due to the assumption that the differentiation lemma already holds for $n$, we can show that for $\stat{n}\in\stopRegionCompl$ it holds that
\begin{align*}
  \left\lvert \frac{\partial }{\partial \costDet_m} \rho_{n-1}(\stat{n})\right\rvert &  = \left\lvert
 p(\Hyp_m)z_n^m\updateProbMeasure[n-1]^m(\stopRegionDec{n}{\bar m}) + \int_{\stopRegionCompl[n]} \rho^\prime_{n,\costDet_m} \dInt\updateProbMeasure[n-1] \right\rvert \\
 & \leq \left\lvert  p(\Hyp_m)z_n^m\updateProbMeasure[n-1]^m(\stopRegionDec{n}{\bar m})\right\vert + \left\lvert\int_{\stopRegionCompl[n]} \rho^\prime_{n,\costDet_m} \dInt\updateProbMeasure[n-1] \right\rvert \\
 & =  p(\Hyp_m)z_n^m\updateProbMeasure[n-1]^m(\stopRegionDec{n}{\bar m}) + \left\lvert\int_{\stopRegionCompl[n]} \rho^\prime_{n,\costDet_m} \dInt\updateProbMeasure[n-1] \right\rvert \\
 & =  p(\Hyp_m)z_n^m\updateProbMeasure[n-1]^m(\stopRegionDec{n}{\bar m}) + \int_{\stopRegionCompl[n]} h_n^m \dInt\updateProbMeasure[n-1]\,.
\end{align*}
Hence, the derivative is bounded on $\stopRegionCompl$ as well. For the induction basis $\rho_N$ it holds that
\begin{align*}
 \rho_{N,\costDet_m}^\prime = \ind{\stopRegionDec{N}{\bar m}} p(\Hyp_m\given\stat{N}) < \infty\,,
\end{align*}
and therefore
\begin{align*}
 \left\lvert\rho_{N,\costDet_m}^\prime\right\rvert = \ind{\stopRegionDec{N}{\bar m}} p(\Hyp_m\given\stat{N}) = h_N^m(\stat{n})\,.
\end{align*}
This concludes the proof. \hfill\qed

\section{Proof of Theorem \ref{theo:derivativesPerformanceMeasures}} \label{app:proofDerivativesPerformanceMeasures}
Assuming that the optimal policy as stated in \cref{corr:optPolicy} is used, the detection errors can be written as:
\begin{align}\label{eq:DetErrPiecewise}
 \errorDet{n}{m}(\stat{n}) & = \begin{cases}
                                0 & \text{for} \quad \stat{n} \in \stopRegionDec{n}{m} \\
                                1 & \text{for} \quad \stat{n} \in \stopRegionDec{n}{\bar{m}}\\
                                \E[\errorDet{n+1}{m}(\stat{n+1})\given\Hyp_m,\stat{n}] & \text{for} \quad \stat{n} \in \stopRegionCompl
                               \end{cases}
\end{align}The expected value in \cref{eq:DetErrPiecewise} can be rewritten as
\begin{align*}
 \E[\errorDet{n+1}{m}(\stat{n+1})\given\Hyp_m,\stat{n}] & = \int \errorDet{n+1}{m} \dInt\updateProbMeasure^m \\
& = \int_{\{\transkernel\in\stopRegion\}}  \sum_{i=1,i\neq m}^M\ind{\transkernel\in\stopRegionDec{n+1}{m}}\dInt\updateProbMeasure^m  + \int_{\{\transkernel\in\stopRegionCompl[n+1]\}}  \errorDet{n+1}{m} \dInt\updateProbMeasure^m \\
& =  \updateProbMeasure^m(\stopRegionDec{n+1}{\bar{m}}) + \int_{\{\transkernel\in\stopRegionCompl[n+1]\}} \errorDet{n+1}{m} \dInt\updateProbMeasure^m\,.
\end{align*}
For $\stat{n}\in\stopRegionDec{n}{m}$, it holds that
\begin{align}
 \frac{\rho_{n,\costDet_m}^\prime(\stat{n})}{p(\Hyp_m)z_n^m} = & \frac{0}{p(\Hyp_m)z_n^m} = 0 = \errorDet{n}{m}(\stat{n})
\end{align}and for $\stat{n}\in\stopRegionDec{n}{\bar{m}}$, it further holds that
\begin{align*}
 \frac{\rho_{n,\costDet_m}^\prime(\stat{n})}{p(\Hyp_m)z_n^m} = & \frac{p(\Hyp_m)z_n^m}{p(\Hyp_m)z_n^m} = 1 = \errorDet{n}{m}(\stat{n})\,,
\end{align*}which is exactly the detection error on the stopping region.
It is now left to show that $\errorDet{n}{m}=\frac{\rho_{n,\costDet_m}^\prime(\stat{n})}{p(\Hyp_m)z_n^m}$ also solves \cref{eq:DetErrPiecewise} on the complement of the stopping region. For $\stat{n}\in\stopRegionCompl$, it holds that
\begin{align*}
 \frac{\rho_{n,\costDet_m}^\prime(\stat{n})}{p(\Hyp_m)z_n^m} & = \updateProbMeasure^m(\stopRegionDec{n+1}{\bar m})+ \int_{\{\transkernel\in\stopRegionCompl[n+1]\}} \frac{\rho_{n+1,\costDet_m}^\prime(\transkernel(\xnew))}{p(\Hyp_m)z_{n+1}^m} p(\xnew\given\Hyp_m,\stat{n}) \dInt\xnew\,, \\
\frac{\rho_{n,\costDet_m}^\prime(\stat{n})}{p(\Hyp_m)z_n^m} & = \updateProbMeasure^m(\stopRegionDec{n+1}{\bar m}) + \!\!\!\!\int\displaylimits_{\{\transkernel\in\stopRegionCompl[n+1]\}} \!\!\!\frac{\rho_{n+1,\costDet_m}^\prime(\transkernel(\xnew))}{p(\Hyp_m)z_{n}^m \frac{p(\xnew\given\Hyp_m,\stat{n})}{p(\xnew\given\stat{n})}} p(\xnew\given\Hyp_m,\stat{n}) \dInt\xnew\,,\\
\frac{\rho_{n,\costDet_m}^\prime(\stat{n})}{p(\Hyp_m)z_n^m} & = \updateProbMeasure^m(\stopRegionDec{n+1}{\bar m}) + \frac{1}{p(\Hyp_m)z_{n}^m} \!\int\displaylimits_{\{\transkernel\in\stopRegionCompl[n+1]\}}\!\! \!\!\!\!\!\!\rho_{n+1,\costDet_m}^\prime(\transkernel(\xnew)) p(\xnew\given\stat{n}) \dInt\xnew\,, \\
\rho_{n,\costDet_m}^\prime(\stat{n}) & = p(\Hyp_m)z_{n}^m\updateProbMeasure^m(\stopRegionDec{n+1}{\bar m}) + \!\!\!\!\int\displaylimits_{\{\transkernel\in\stopRegionCompl[n+1]\}}\!\!\!\rho_{n+1,\costDet_m}^\prime\dInt\updateProbMeasure\,,
 \end{align*}which is true by \cref{lem:derivatives}. Hence, \cref{theo:derivativesPerformanceMeasures} holds for $\errorDet{n}{m}$. \hfill\qed

\section{Proof of Theorem \ref{theo:optCostCoeff}}\label{sec:proofOptCoeff}
It has to be shown that
\begin{align}
 & \max_{\costDet\geq0,\costEst\geq0}\, L_{\detConstr[],\estConstr[]}(\costDet,\costEst) =  \max_{\costDet\geq0,\costEst\geq0} \,\left\{ \rho_0(\stat{0}) \!-\! \sum_{m=1}^M p(\Hyp_m)(\costDet_m\detConstr + \costEst_m\estConstr)\right\} \label{eq:LagrDualApp}
\end{align}
attains its maximum for some non-negative and finite values of $\costDet^\star, \costEst^\star$ and that the solutions of \cref{eq:LagrDualApp} and \cref{eq:constrProblem} coincide.

By applying \cref{lem:derivatives,theo:derivativesPerformanceMeasures}, one obtains
\begin{align*}
\frac{\partial}{\partial \costDet_m}L_{\detConstr[],\estConstr[]}(\costDet,\costEst) & = p(\Hyp_m)(\errorDet{}{m} - \detConstr)\,,\; m=1,\ldots,M\,,\\
\frac{\partial}{\partial \costEst_m}L_{\detConstr[],\estConstr[]}(\costDet,\costEst) & = p(\Hyp_m)(\errorEst{}{m} - \estConstr)\,,\; m=1,\ldots,M\,.
\end{align*}
Since all constraints in \cref{eq:constrProblem} are inequality constraints, $\costDet^\star$/$\costEst^\star$ are solutions of \cref{eq:constrProblem} if $\costDet^\star$/$\costEst^\star$ are positive and the corresponding derivative vanishes, or if $\costDet^\star$/$\costEst^\star$ are zero and the corresponding derivative is negative. The first case, i.e., when the gradient vanishes, holds, if

\begin{align*}
 \errorDet{}{m} & = \detConstr \quad\text{and}\quad \errorEst{}{m}  = \estConstr\,,\quad m=1,\ldots,M\,,
\end{align*}i.e., the constraints are fulfilled with equality. It now has to be shown that these $\costDet^\star,\costEst^\star$ are non-negative and finite. This is only outlined for $\costDet^\star$, since it can be shown similarly for $\costEst^\star$. We first consider the limit
\begin{align*}
\lim_{\costDet_m\rightarrow\infty}\frac{\partial}{\partial \costDet_m}L_{\detConstr[],\estConstr[]}(\costDet,\costEst) & = p(\Hyp_m)(0 - \detConstr)<0\,,
\end{align*}
which contradicts the fact that an infinitely large $\costDet_m$ is a solution of \cref{eq:LagrDualApp}, since a negative gradient would only result in a maximum if and only if  $\costDet^\star_m=0$. Next, the gradient at $\costDet_m^\star=0$ needs closer inspection, i.e.,
\begin{align*}
\frac{\partial }{\partial \costDet_m}L_{\detConstr[],\estConstr[]}(\costDet,\costEst)\bigg\rvert_{\costDet_m^\star=0} = p(\Hyp_m)(\errorDet{}{m} - \detConstr)\,.
\end{align*}
At this point, it has to be mentioned again that the cost for rejecting hypothesis $\Hyp_m$ does not only depend on $\costDet_m$ but rather on all $\costDet_i$ as well as on $\costEst_j$, $j\in\{1,\ldots,M\}\setminus m$. Hence, even if $\costDet^\star_m=0$, i.e., the constraint on the error probability under $\Hyp_m$ is not enforced with equality, it can still be satisfied implicitly, as a consequence of the remaining constraints.
We distinguish the two cases in which the resulting detection error is smaller and larger than the error constraint $\detConstr$, respectively.
In the former case, the gradient is negative. This results, in combination with $\costDet_m^\star=0$, in an optimum (complementary slackness).
In the case where the resulting detection error is larger than the constraint, the gradient becomes positive. This contradicts the assumption that $\costDet_m^\star$ is an optimal solution.
Due to the concavity of $L_{\detConstr[],\estConstr[]}(\costDet,\costEst)$ and the fact that an infinitely large $\costDet_m$ results in a negative gradient, a positive gradient for $\costDet_m=0$ implies that there exists a positive and finite $\costDet_m^\star$ such that the gradient vanishes.
That is, the designed sequential scheme fulfills the requirements and is of minimum run-length by definition.
It can now be easily shown that the optimal objective is the average run-length
\begin{align*}
  L_{\detConstr[],\estConstr[]}(\costDet^\star,\costEst^\star) & =  \rho_0(\stat{0}) - \sum_{m=1}^M p(\Hyp_m)(\costDet_m^\star\detConstr + \costEst_m^\star\estConstr) \\
& = \E[\tau] + \sum_{m=1}^M p(\Hyp_m)(\costDet_m^\star\errorDet{}{m} + \costEst_m^\star\errorEst{}{m}) - \sum_{m=1}^M p(\Hyp_m)(\costDet_m^\star\detConstr + \costEst_m^\star\estConstr)\\
  & = \E[\tau]\,.
\end{align*}
This concludes the proof. \hfill\qed

\section{Derivation of the Posterior Probabilities of the Hypotheses} \label{app:derPostH}
According to Bayes' theorem, we can calculate the posterior probability as
\begin{align*}
 p(\Hyp_m\given\stat{n}) = \frac{p(\stat{n}\given\Hyp_m)p(\Hyp_m)}{p(\stat{n})}\,.
\end{align*}The marginal density is given by
\begin{align}
 p(\stat{n}\given\Hyp_m) & = \int p(\stat{n}\given\Hyp_m,\sigma^2)p(\sigma^2)\dInt\sigma^2 \nonumber\\
			 \begin{split}
			    & =  \int \bigl(2\pi\sigma^2\bigr)^{-\frac{n}{2}}\!\exp\biggl(\!\!-\frac{\sum_{k=1}^n(x_k-A_m)^2}{2\sigma^2}\biggr) \exp\biggl(-\frac{\scale}{\sigma^2}\biggr)\frac{\scale^\shape}{\Gamma(\shape)}  \dInt\sigma^2\,.
			     \end{split}
\end{align}Using the parametrization given in \cref{eq:paramPostShape,eq:paramPostScale}, we obtain
\begin{align*}
 p(\stat{n}\given\Hyp_m) &=   \bigl(2\pi\bigr)^{-\frac{n}{2}} \frac{\scale^\shape}{(\scalePost)^\shapePost} \frac{\Gamma(\shapePost)}{\Gamma(\shape)} \int \invGam(\sigma^2\given\shapePost,\scalePost)\dInt\sigma^2 \\
                          &=  \bigl(2\pi\bigr)^{-\frac{n}{2}} \frac{\scale^\shape}{(\scalePost)^\shapePost} \frac{\Gamma(\shapePost)}{\Gamma(\shape)}\,.
\end{align*}
Hence, the posterior probabilities are given by
\begin{align}
 p(\Hyp_m\given\stat{n}) =  K p(\Hyp_m) \bigl(2\pi\bigr)^{-\frac{n}{2}} \frac{\scale^\shape}{(\scalePost)^\shapePost} \frac{\Gamma(\shapePost)}{\Gamma(\shape)}\label{eq:post}\,,
\end{align}
where the normalization constant $K$ is given by
\begin{align}
 K=\left(\sum_{m=1}^M p(\Hyp_m\given\stat{n})\right)^{-1} \label{eq:postNormConst}\,.
\end{align}

\section{Derivation of the Posterior Predictive} \label{app:derPostPred}
The posterior predictive can be factorized as
\begin{align}\label{eq:postPredGen}
 p(\xnew\given\stat{n}) = \sum_{m=1}^M p(\xnew\given\Hyp_m,\stat{n})p(\Hyp_m\given\stat{n})\,.
\end{align}The conditional posterior predictive can now by calculated as
\begin{align*}
 p(\xnew\given\Hyp_m,\stat{n}) & = \int p(\xnew\given\Hyp_m,\sigma^2) p(\sigma^2\given\Hyp_m,\stat{n})\dInt\sigma^2\\
 & = \int \frac{(\scalePost)^{\shapePost}}{\Gamma(\shapePost)} (\sigma^2)^{-\shapePost - 1}\exp\biggl(-\frac{\scalePost}{\sigma^2}\biggr) (2\pi\sigma^2)^{-0.5}\exp\left(-\frac{(\xnew-A_m)^2}{2\sigma^2}\right)\dInt\sigma^2 \\
 & = \frac{1}{\sqrt{2\pi}} \frac{(\scalePost)^\shapePost}{\Gamma(\shapePost)} \int  (\sigma^2)^{-\shapePost - 1.5}\exp\biggl(-\frac{\scalePost + 0.5(\xnew-A_m)^2}{\sigma^2}\biggr) \dInt\sigma^2\,.
\end{align*}
Substituting
\begin{align*}
 \shapePP & = \shapePost + 0.5 \quad \text{and} \quad
 \scalePP  = \scalePost + 0.5(\xnew-A_m)^2\,,
\end{align*}we can rewrite the conditional posterior predictive as
\begin{align}
\begin{split}\label{eq:condPostPredParam}
 p(\xnew\given\Hyp_m,\stat{n}) & = \frac{1}{\sqrt{2\pi}} \frac{(\scalePost)^\shapePost}{\Gamma(\shapePost)} \\
  & \times \int  (\sigma^2)^{-\shapePP - 1}\exp\left(-\frac{\scalePP}{\sigma^2}\right) \dInt\sigma^2 \,.
\end{split}
\end{align}
As the integrand in \cref{eq:condPostPredParam} is the scaled \ac{pdf} of the inverse Gamma distribution, we simplify \cref{eq:condPostPredParam} to
\begin{align}
p(\xnew\given\Hyp_m,\stat{n}) & = \frac{1}{\sqrt{2\pi}} \frac{(\scalePost)^\shapePost}{\Gamma(\shapePost)} \frac{\Gamma(\shapePP)}{(\scalePP)^\shapePP} \label{eq:condPostPred}\,.
\end{align}
Inserting \cref{eq:condPostPred} and \cref{eq:post} into \cref{eq:postPredGen} and using $K$ as defined in \cref{eq:postNormConst}, the posterior predictive is given by
\begin{align*}
 p(\xnew\given\stat{n}) = K \bigl(2\pi\bigr)^{-\frac{n+1}{2}} \sum_{m=1}^M p(\Hyp_m) \frac{\Gamma(\shapePP)}{\Gamma(\shape)}  \frac{\scale^\shape}{(\scalePP)^\shapePP}\,.
\end{align*}

\section*{Acknowledgements}
The authors would like to thank the anonymous reviewers and the editors for their time and effort.
The work of Dominik Reinhard was supported by the German Research Foundation (DFG) under grant number 390542458.
The work of Michael Fau\ss{} was supported by the German Research Foundation (DFG) under grant number 424522268.
 
\bibliographystyle{tandfx}
\setcitestyle{authoryear, open={(},close={)}}
\bibliography{IEEEabrv,mrabbrev,references,applications} \end{document}